\documentclass[11pt]{article}

\usepackage[margin=1in]{geometry}
\usepackage{lmodern}
\usepackage[T1]{fontenc}
\usepackage[utf8]{inputenc}
\usepackage[tbtags]{amsmath}
\allowdisplaybreaks
\usepackage{amssymb,amsthm,mathtools,bbm}
\usepackage{physics}
\usepackage{hyperref}
\usepackage[svgnames]{xcolor}
\usepackage[capitalise,nameinlink]{cleveref}
\definecolor{links}{RGB}{11, 85, 255}
\definecolor{cites}{RGB}{11, 85, 255}
\definecolor{urls}{RGB}{255, 116, 0}
\hypersetup{colorlinks={true},linkcolor={links},citecolor=[named]{cites},urlcolor={urls}}
\usepackage[compress]{natbib} 
\usepackage{pgfplots}
\pgfplotsset{compat=1.14}
\usepgfplotslibrary{fillbetween}
\usepackage{hyphenat}
\usepackage{caption}
\usepackage{subcaption}
\usepackage{appendix}
\usepackage{nicefrac}
\usepackage{tcolorbox}
\usepackage[symbol]{footmisc}
\usepackage{float}
\usepackage[section]{placeins}

\usepackage{colortbl} 
\usepackage{booktabs} 
\usetikzlibrary{calc}

\usepackage{thm-restate}

\usepackage{algorithm}
\usepackage[noend]{algpseudocode}

\usepackage{enumitem}

\newcommand{\cA}{\mathcal{A}}
\newcommand{\R}{\mathbb{R}}
\newcommand{\Z}{\mathbb{Z}}

\newcommand{\B}{\mathbb{B}}

\newcommand{\cI}{\mathcal{I}}
\newcommand{\cF}{\mathcal{F}}

\DeclareMathOperator{\vol}{Vol}
\DeclareMathOperator{\width}{width}
\DeclareMathOperator{\proj}{proj}
\renewcommand{\ip}[1]{\langle #1 \rangle}
\newcommand{\E}[1]{\mathbb{E}\left[#1\right]}

\newcommand{\ind}[1]{\mathbbm{1}\left\{#1\right\}}
\newcommand{\eps}{\varepsilon}

\newcommand{\gft}{\textup{\textsc{GFT}}}
\newcommand{\profit}{\textup{\textsc{Profit}}}

\newcommand{\itb}{\proj(B_t,x_t)}
\newcommand{\its}{\proj(S_t,x_t)}

\newcommand{\us}{\overline{s}_t}
\newcommand{\ls}{\underline{s}_t}

\newcommand{\ub}{\overline{b}_t}
\newcommand{\lb}{\underline{b}_t}

\newtheorem{theorem}{Theorem}[section]
\newtheorem*{theorem*}{Theorem}
\newtheorem{lemma}[theorem]{Lemma}

\newtheorem{observation}[theorem]{Observation}
\newtheorem{fact}[theorem]{Fact}

\begin{document}

\title{Contextual Online Bilateral Trade\footnote{Part of this work was done while Federico Fusco, Stefano Leonardi, and Matteo Russo were visiting Anupam Gupta at New York University. The work of RC and AG was supported in part by NSF awards CCF-2422926 and CCF-2608359. The work of FF and SL was supported in part by the MUR PRIN grant 2022EKNE5K and the FAIR (Future Artificial Intelligence Research) project PE0000013,
funded by the NextGenerationEU program within the PNRR-PE-AI scheme (M4C2, investment
1.3, line on Artificial Intelligence)}}

\date{}
\author{}

\author{Romain Cosson\thanks{Computer Science Department, New York University.} \and  Federico Fusco\thanks{Sapienza University, Rome, Italy} \and Anupam Gupta\footnotemark[2]
\and Stefano Leonardi\footnotemark[3] \and Renato Paes Leme\thanks{Google Research, NYC, United States} \and Matteo Russo\thanks{EPFL, Lausanne, Switzerland}}

\maketitle
\thispagestyle{empty}

\begin{abstract}

    We study repeated bilateral trade when the valuations of the sellers and the buyers are contextual. More precisely, the agents' valuations are given by the inner product of a context vector with two unknown $d$-dimensional vectors---one for the buyers and one for the sellers.

    At each time step $t$, the learner receives a context and posts two prices, one for the seller and one for the buyer, and the trade happens if both agents accept their price. We study two objectives for this problem, gain from trade and profit, proving no-regret with respect to a surprisingly strong benchmark: the best omniscient dynamic strategy.

    In the natural scenario where the learner observes \emph{separately} whether the agents accept their price---the so-called \emph{two-bit} feedback---we design algorithms that achieve $O(d\log d)$ regret for gain from trade, and $O(d  \log\log T + d\log d)$ regret for profit maximization. Both results are tight, up to the $\log(d)$ factor, and implement per-step budget balance, meaning that the learner never incurs negative profit.

    In the less informative \emph{one-bit} feedback model, the learner only observes whether a trade happens or not. For this scenario, we show that the tight two-bit regret regimes are still attainable, at the cost of allowing the learner to possibly incur a small negative profit of order $O(d\log d)$, which is notably independent of the time horizon. As a final set of results, we investigate the combination of one-bit feedback and per-step budget balance. There, we design an algorithm for gain from trade that suffers regret independent of the time horizon, but \emph{exponential} in the dimension $d$. For profit maximization, we maintain this exponential dependence on the dimension, which gets multiplied by a $\log T$ factor.
    
\end{abstract}

\newpage
\pagenumbering{arabic} 
\section{Introduction}
\label{sec:introduction}
    Online pricing is a well-studied problem at the intersection of theoretical computer science, machine learning, and economics \citep{LiuLS21,LemeSTW23,SinglaW24} that dates back to the seminal papers by \citet{KleinbergL03} and \citet{BlumKRW03}. The task consists of finding the right pricing rule for certain types of goods by repeatedly interacting with prospective buyers and observing their reactions to the prices posted. A natural extension consists in studying the situation where the mechanism designer does not directly possess the goods to trade, but acts as an intermediary between demand and supply. While the offline version of this \emph{bilateral trade} problem dates back to classic works of \citet{Vickrey61} and \citet{MyersonS81}, its online version has received a rapidly growing interest within the computer science community \citep{Cesa-BianchiCCF24,Cesa-BianchiCCF24a,AzarFF24,BernasconiCCF24,LunghiCM26}. Indeed, finding the right balance between sellers and buyers gives rise to new and exciting challenges that qualitatively separate one-sided pricing and bilateral trade. In this work, we investigate bilateral trade from an online learning perspective, assuming that sellers' and buyers' valuations are parameterized as functions of context vectors \citep{GaucherBCCP25}.
    
    Consider the following scenario, where an internet platform repeatedly intermediates between supply and demand in a certain market (e.g., ride-sharing, e-commerce, or ads). At each time step $t$, a new seller joins the platform, interested in selling a certain good; at the same time, an interested buyer arrives, and the platform observes their features summarized in a $d$-dimensional context vector $x_t$. The agents' private valuations for the good for sale are then given by the public features, averaged by unknown but fixed \emph{weights}. Formally, we assume that there are two unknown vectors $s$ and $b$ in $\R^d$ that represent the relative importance of the various features, and that the agents' valuations $s_t = \langle s, x_t \rangle$ and $b_t = \langle b, x_t \rangle$ are given by the scalar product of these two ``ground truths'' and the public context. Given the context $x_t$, the platform proposes two prices, $p_t$ to the seller and $q_t$ to the buyer, and the trade occurs if and only if both agents accept their prices. 

    From a learning perspective, the platform's goal is to design a pricing policy that takes as input the context at time $t$ and computes the right prices based on agents' past behavior. In this work, we consider the two most natural performance objectives for this problem: \emph{economic efficiency}, measured in terms of the gain from trade, and \emph{profit}. 
    \begin{itemize}
        \item The \textbf{gain from trade} equals the difference $b_t - s_t$ if the trade happens, and zero otherwise.
        \item The \textbf{profit} is simply the difference between the price paid by the buyer and that paid to the seller, $ q_t-p_t$ if the trade happens, and zero otherwise.
    \end{itemize}
    At a high level, the gain from trade is maximized as soon as the learner posts two prices contained in the interval $[s_t,b_t]$, while to extract good profit, the prices should also be ``close'' to the corresponding valuations. Our goal is to design learning algorithms for these two problems that exhibit small regret with respect to an ambitious benchmark: the best omniscient pricing strategy that knows $s$ and $b$ in advance. Notably, there exists a strategy that is optimal for both efficiency and profit maximization, which consists of posting $p_t = s_t$ and $q_t = b_t$ when $s_t \le b_t$. Therefore, the two metrics share the same benchmark: $\sum_t (b_t-s_t)_+$, where $(\cdot)_+$ denotes the positive part.  
    
    Following the literature \citep[e.g.,][and follow-ups]{Cesa-BianchiCCF24}, we consider two realistic feedback models for repeated bilateral trade: one- and two-bit feedback.
    \begin{itemize}
        \item \textbf{two-bit}: The learner observes the willingness to accept their prices for both agents.
        \item \textbf{one-bit}: The learner only observes whether the trade happened or not.
    \end{itemize}
    These two models are motivated by practical considerations. The two-bit setting models the situation in which the learner communicates prices to the agents and observes whether each agent individually agrees to trade. Conversely, in the one-bit setting, the learner does not receive the disaggregated information on who rejected its proposed price. Furthermore these two feedback structures exhibit a desirable property that characterizes take-it-or-leave-it prices: agents communicate only a single bit of information to the platform, without revealing their private information.

    Beyond the non-trivial interplay of the two agents' valuations in designing good pricing policies, bilateral trade exhibits a qualitative difference with respect to simple (one-sided) pricing: the requirement for the mechanism designer not to run into a deficit to make the trade happen. This is formalized by the \textbf{budget-balance} property \citep{MyersonS81}, which prescribes some constraint on how much the learner can \emph{subsidize} the market (i.e., incur negative profit).
    
    \subsection{Our Results}
    \label{subsec:results}

        We present here our results; we also refer to Table~\ref{tab:results} for a visual summary.

        \paragraph{Two-bit Feedback.}In the two-bit feedback model, we have the following results:
        \begin{itemize}
            \item We construct an \emph{efficiency-maximizing}
              algorithm that suffers $O(d \log d)$ regret
              (\Cref{thm:2-bit-gft}). Our rate removes the dependence
              on $T$ and improves on the dependence on $d$ over the previous
              known bound of $\tilde{O}(d^2 \log T)$\cite{GaucherBCCP25}; this is only a
              logarithmic (in $d$) factor 
              from the lower bound of $\Omega(d)$
              (\Cref{thm:2-bit-gft-lower}).
            \item We construct a \emph{profit-maximizing} algorithm
              that suffers $O(d \log \log T + d \log d)$ regret
              (\Cref{thm:2-bit-profit}). We show that the
              doubly-logarithmic dependence in $T$ is essentially
              tight, by adapting ideas from \citet{KleinbergL03} to
              get lower bound of $\Omega(d\log \log \nicefrac{T}d)$
              (\citet{LiuLS21};
              see also \Cref{thm:2-bit-profit-lower}).
        \end{itemize}

        The above algorithms are deterministic and
        \emph{per-round budget balanced} \citep[][and
        follow-ups]{Cesa-BianchiCCF24}: they only post prices
        $p_t \le q_t$, ensuring that the per-round profit is non-negative.

        \paragraph{One-bit Feedback.}In the one-bit feedback model, we
        present two sets of results, for different notions of budget
        balance.  Our first set of results is when we allow ourselves
        the flexibility of posting prices $p_t > q_t$: 
        \begin{itemize}
        \item We adapt the two-bit feedback algorithms for
          both gain-from-trade and profit to work in the one-bit
          feedback model, at the expense of slightly violating
          \emph{per-round} budget balance. 
          \item In particular, we recover the two-bit regret rates, with an
          overall negative profit\footnote{This is referred to as
            ``budget violation'' in the recent work by
            \citet{LunghiCM26}.} of order $O(d \log d)$
          (\Cref{thm:gft-safe,thm:profit-safe}). Notably, this term is independent of the time horizon.
        \end{itemize}
        We then consider per-round budget balance
        (i.e., where we post prices $p_t \leq q_t$), and design randomized learning algorithms enforcing such property with the following regret rates:
        \begin{itemize}
        \item We construct an \emph{efficiency-maximizing} algorithm
          that exhibits a regret independent of $T$ of order
          $O(d 6^d)$ (\Cref{thm:one-bit-gft}).
        \item We construct a \emph{profit-maximizing} algorithm with
          regret that depends logarithmically on $T$, of order
          $O(d 6^d \log T)$ (\Cref{thm:one-bit-profit}).
        \end{itemize}

        \begin{table}[t!]
        \centering
        \renewcommand{\arraystretch}{1.5}
        \begin{tabular}{l|l|l|l|}
            \cline{2-4}
                                                  & Two-bit / BB                        & One-bit / non-BB                      & One-bit / BB \\ \hline
            \multicolumn{1}{|l|}{$\gft$} & $\tilde \Theta(d)$             & $\tilde \Theta(d)$                     & $O(d6^d)$         \\ \hline
            \multicolumn{1}{|l|}{$\profit$}          & $\tilde \Theta(d\log\log T)$ & $\tilde \Theta(d\log\log T)$ & $O(d6^d \log T)$  \\ \hline
            \end{tabular}
            \caption{\footnotesize Summary of our results. The $\tilde{\Theta}$ hides factors that are logarithmic in the dimension. The results in the middle column hold at the expense of an overall negative profit of order $O(d \log d)$. BB stands for (per-round) budget balance. GFT stands for gain from trade.}
            \label{tab:results}
        \end{table}

        While explicitly requiring budget balance is crucial for efficiency maximization \citep{MyersonS81} (we do not want the learner to subsidize the market), it is less prominent for profit maximization, where the learner's goal of accruing profit is naturally aligned with the budget balance property.  In any case, budget balance is a desirable and natural property also for a profit-maximizing algorithm, and we investigate its impact on both metrics for completeness. 
        Overall, it is remarkable that achieving such a small regret rate (even time-independent for $\gft$) against the best dynamic prices is possible, in a non-stationary (although structured) environment.

        \paragraph{Comparison with previous work.} While we are the first to study profit maximization,
        efficiency maximization in contextual bilateral trade has been
        introduced by \citet{GaucherBCCP25}. They provide an $\tilde
        O(d^2\log T)$ regret result for two-bit feedback, which we
        improve to a nearly-tight bound of $O(d \log d)$,
        removing the dependence on $T$ and a $\tilde
        \Theta(d)$ factor. \citet{GaucherBCCP25} also investigate
        one-bit feedback, constructing algorithms that violate
        per-round budget balance, but have overall non-negative profit
        (This is called the global budget balance condition in \citet{BernasconiCCF24}). However, such results
        are instance-dependent and exhibit regret rates that are
        polynomial in $T.$ In contrast, we either allow for
        (extremely) small negative profit and obtain tight regret
        rates, or maintain per-round budget balance, at the cost of
        worse dependence on the dimension $d$.
        
        \paragraph{Future Directions.} We provide a fairly complete understanding of contextual bilateral trade under different feedback models and budget balance constraints. We mention two open directions for future works. (i) Characterizing the ``budget vs. regret'' trade off for one-bit feedback algorithms: given a violation budget $B=B(d,T)$, what is the minimax regret rate achievable with one-bit feedback and allowing for budget violation $B$ (similarly to what has been recently done for adversarial bilateral trade in \citet{LunghiCM26})? (ii) Investigating contextual 
        models that allow noisy or imprecise feedback. While more natural in applications, these models may not allow for sublinear regret against the best omniscient pricing strategy, so that weaker benchmarks (e.g., ``best-fixed-strategy'') may need to be considered, likely with polynomial regret rates (as in \citet{GaucherBCCP25}).

    \subsection{Technical Challenges and Our Approach}
    \label{subsec:challenges}

    {The algorithmic problem we face has a clean geometric
      interpretation: we want to $s, b \in \R^d$ using
      special types of probes: given context
      $x \in \R^d$, we can post two prices $p$ and $q$, and observes
      the outcome of the trade---whether $\langle s, x \rangle \leq p$
      \emph{and} $\langle b, x \rangle \geq p$. E.g., if the seller
      accepts its price, but the trade fails because the buyer refuses
      $q$, we know that $s$ lies in the half-space
      $\{\tilde s\mid\ip{\tilde s,x} \ge p \}$ and
      $b \in \{\tilde b\mid \ip{\tilde b,x} \le q \}$. The learning
      algorithm can then maintain, for each agent, a ``confidence
      region'' defined by the intersections of these half-spaces,
      given the feedback thus far (see also
      \Cref{fig:visualization-intro}.)}

    \paragraph{Potentials and regret.} 
    A good algorithm for contextual bilateral trade must find the right balance between
      narrowing these confidence regions, and suffering as little
      regret as possible. 
      To capture this, it is natural to define a notion of
      ``potential'' which quantifies such delicate trade-off. In
      particular, our starting point is the family of potential
      functions based on Steiner polynomials introduced by
      \citet{LiuLS21}.

    {For the bilateral trade problem, we need to
      consider two potentials, one corresponding to the seller's
      confidence region, and another to the buyer. Crucially, we
      cannot handle the two 
      separately: 
      (i) The prices must respect budget balance, so the choice of
      price $p$ restricts the values that $q$ can take, and (ii) both
      the gain-from-trade, and the profit depend on \emph{both} agents
      accepting their price. E.g., a good price for the seller may
      still incur constant regret if the buyer rejects the trade. To
      strike the right balance, we consider projections of the
      confidence regions onto the current context: if the two
      projections are disjoint, we can consider the two agents
      separately; otherwise, we focus on the agents corresponding to
      larger projection (i.e., with greater uncertainty along the
      current context), and use the length of its projection as a
      proxy for the incurred regret.  }

\paragraph{Ambiguity in one-bit feedback.} 
Updating the confidence regions for the buyer and seller separately is
only feasible if we have explicit feedback from both agents. In the
one-bit setting, a ``no-trade'' outcome is ambiguous: the learner
cannot discern whether the buyer's price was too high or the seller's
price was too low (or both). Indeed, using this feedback directly gives us
non-convex uncertainty regions, which are difficult to handle. To overcome this ambiguity, we give two different solutions,
  each characterized a different notion of budget balance. 
  
  In the
  first set of results, we show how to ``simulate'' two-bit feedback
  by posting a ``safe'' price for at least one of the two agents;
  i.e., a price that will surely be accepted. To do this, the learner may
  post a price to the buyer which is below that for the seller,
  thereby violating budget balance; hence, we need to carefully bound
  the negative profit incurred. In the second set of results which
handle this ambiguity without violating the budget balance, we update
the confidence regions \emph{only} when a trade successfully occurs,
as this is the only situation where the agents' hidden bits are
unambiguous. However, this creates a new problem: discarding updates
from no-trade rounds breaks the analysis used in two-bit algorithms
(notably, the two-bit profit-maximizing algorithm typically relies on
decreasing the potential doubly exponentially when no trade is
observed).

To resolve this, we introduce a novel randomized pricing strategy. We
demonstrate that the algorithm incurs constant expected regret between
periods of successful trade. This allows us to ``charge'' the regret
accumulated during no-trade steps to the time steps where a trade
occurs and the potential decreases. Because the prices inducing this
potential decrease are randomized, they result in a slower reduction
of potentials (as opposed to the deterministic ones allowed in the
two-bit feedback), yielding a looser dependence on the dimension.

    \subsection{Related Work}
    \label{subsec:related}

    Our work is at the intersection of three main lines of research: online contextual pricing, the study of efficiency and profit maximization in bilateral trade, and its online learning version. 

    \paragraph{Online Contextual Pricing.} Our setting generalizes the (one-sided) contextual pricing problem introduced by \citet{AminRS14}. A sequence of papers \citep{CohenLL20,LemeS22,LiuLS21} introduced techniques from convex geometry to obtain increasingly better bounds for this problem; the state of the art is the $O(d \log \log T)$ algorithm in \citet{LiuLS21}, which is essentially tight \citep{KleinbergL03,LiuLS21}. This problem has also been studied under stochastic noise models \citep{JavanmardN19,Javanmard17,ShahJB19,XuW21,XuW22}, adversarial noise models \citep{KrishnamurthyLPS23,LemePS22} and different assumptions on the valuation \citep{MaoLS18}.  

    \paragraph{(Bayesian) Bilateral Trade.} In the Bayesian Bilateral Trade problem, the agents' private valuations are drawn from some \emph{known} distributions, and the goal is to design trading mechanisms that are robust to the strategic behavior of the agents. The most studied objective is efficiency maximization, measured in terms of the multiplicative gap between the social welfare or gain from trade achieved and the best possible attainable \citep[e.g.,][]{Colini-Baldeschi17a,BlumrosenD21,DuttingFLLR21,DengMSW22}. While social welfare and gain from trade are equivalent in a regret-minimization setting (social welfare is simply gain from trade plus seller's valuation), obtaining a multiplicative approximation for the latter turned out to be far more challenging. The structure of the profit-maximizing mechanism was already well understood in the classical paper by \citet{MyersonS81}, and recently there has been a systematic study of the profit-efficiency trade-off attainable with truthful mechanisms \citep{HajiaghayiHPS25}.

    \paragraph{Online Learning and Bilateral Trade.} A long line of
    research has investigated bilateral trade from the online learning
    perspective. The main difference with our setting is that, in
    those settings, the valuations are generated by an adversary (either i.i.d.~stochastic, $\sigma$-smooth, or adversarial), without a notion of ``context''. The regret is then measured with respect to the best fixed price in hindsight (as opposed to the best dynamic pricing policy), and the rates achievable are typically polynomial in $T$. These differences effectively change the nature of the problem (and the techniques one may use): while in our problem we are trying to locate the $(s,b)$ pair, in online learning the goal is more on ``hedging'' between all possible prices to be competitive with the one that performs better on adversarial or stochastic inputs. For gain-from-trade maximization, tight regret rates are known for various combinations of feedback models, data-generation assumptions, and versions of budget balance constraints \citep{Cesa-BianchiCCF24,Cesa-BianchiCCF24a,AzarFF24,BernasconiCCF24,LunghiCM26}. We also mention works 
    generalizing the study of bilateral trade to 
    to settings with more than one seller and one buyer at each time step \citep{BabaioffFN24}, or 
    to the convergence of natural pricing behaviors \citep{DengMS0W25}. For profit maximization, a recent paper has investigated the online learnability of profit-maximizing mechanisms beyond fixed-price \citep{DiGregorioDFS25}. 

    \paragraph{Independent and concurrent work.} In an independent and concurrent work, \citet{CocciaBC26} study another contextual model for bilateral trade, where the functions mapping contexts to valuations are Lipschitz. While their model is more general than ours, none of our results is subsumed by theirs: they prove regret bounds which exhibit \emph{polynomial} dependence in the time horizon $T$, which are tight for general Lipschitz functions, but not for the linear case we study. Given the nature of their rates, they focus on the case where the time horizon is way larger than the dimension $d$ of the problem, i.e. $T \gg 2^d$. Finally, they only investigate gain-from-trade maximization, and not profit. 

\section{Model and Preliminaries}
\label{sec:model}

    In contextual bilateral trade, a learner repeatedly interacts with a sequence of seller-buyer pairs 
    with the goal of maximizing either the economic efficiency, measured in terms of the gain from trade, or its own profit. The learning protocol is as follows: At each time $t=1, \dots, T$, the learner observes a context $x_t$ in the unit ball $\B = \{w \in \R^d \mid \|w\|_2 \leq 1\}$ and posts two prices, $p_t$ to the seller and $q_t$ to the buyer. The trade happens if and only if both agents accept their prices, i.e., if $p_t$ is at least the seller's valuation $s_t$, and $q_t$ is at most the buyer's valuation $b_t$. The learner then observes feedback that depends on the posted prices and the agents' valuations. We assume that $s_t$ and $b_t$ follow a linear structure\footnote{Throughout the paper, we present our results in terms of linear functions, but the model generalizes to any parameterized function of the type $f(x) = \langle w, \phi(x)\rangle$ for an arbitrary transformation of the context vectors. For example, polynomials of degree $k$ 
    can be written as linear functions of
    $\phi(x) = (\prod_i x_i^{s_i})_{s \in \Z_+^d; \Vert s \Vert_1 \leq k}$. Thus, our bounds extend to polynomials of degree $k$, replacing $d$ with $O(d^k)$.}, namely that there exist two vectors $s,b \in \B$ fixed but unknown to the learner such that $s_t = \ip{s,x_t}$, and $b_t = \ip{b,x_t}$. 
   
    We study two feedback models: (i) two bits, where the learner observes $(\ind{s_t \le p_t}, \ind{q_t \le b_t})$, and (ii) one bit, where it only gets $ \ind{s_t \le p_t}\cdot \ind{q_t\le b_t}$. In the former model, the learner has access to the relative position of both agents' valuations and their posted prices, while in the latter, it only observes whether the trade happened. We consider the gain from trade and profit objectives for the learner, as introduced in \Cref{sec:introduction}. At time $t$, we have (in explicit form):
    \begin{align*}
        \gft_t(p_t,q_t) &= \ip{b-s,x_t}\cdot \ind{\ip{s,x_t} \le p_t}\cdot\ind{q_t \le \ip{b,x_t}}\\
        \profit_t(p_t,q_t) &= (q_t-p_t) \cdot \ind{\ip{s,x_t} \le p_t}\cdot\ind{q_t \le \ip{b,x_t}}.
    \end{align*}

    \begin{observation}
        Given the structure of the agents' valuations and the definitions of gain from trade and profit, we can assume without loss of generality that the contexts $x_t$ have unitary norm. We can also consider other norms beyond the $\ell_2$ one, for instance, we may also assume $x_t,s,$ and $b$ bounded according to other norms (e.g., $\|x_t\|_\infty \le 1$ and $\|s\|_1,\|b\|_1\le 1$), and the results do not change significantly, as long as profit and gain from trade remain bounded.
    \end{observation}
    
    We measure the performance of learning algorithms in terms of their regret with respect to the optimal pricing policy in hindsight. Given a possibly randomized efficiency-maximizing algorithm $\cA$, and an instance $\cI$ of the problem, we define the regret of $\cA$ against $\cI$:
    \[
        R_T(\cA,\cI) = \sum_{t=1}^T (b_t - s_t)_+ - \sum_{t=1}^T\gft_t(p_t,q_t),
    \]
    where the algorithm posts prices $p_t$ and $q_t$, and the instance is characterized by the fixed vectors $(s,b)$, and an arbitrary sequence of contexts. The regret for profit-maximizing algorithms is analogous, with $\profit_t$ instead of $\gft_t$. We want a learning algorithm with small \emph{worst-case} regret: $R_T(\cA) = \sup_{\cI}\E{R_T(\cA,\cI)}$, where the expectation is with respect to the randomness of the algorithm; note, we consider an oblivious adversary, which selects $\cI$ a priori. $\gft_t$ and $\profit_t$ map to $[-2,2]$, so the typical goal is to obtain regret that is sublinear in the time horizon $T$, which implies that the performance of the learning algorithm converges \emph{on average} to the benchmark.

    \paragraph{Budget Balance.} We consider two notions of budget balance: 
    a per-step notion, where  $p_t \le q_t$ for every $t$, and a weaker one, where we require that the overall negative profit is small. In our results (see \Cref{thm:gft-safe,thm:profit-safe}) we enforce that  $\sum_{t=1}^T \max\{0,-\profit_t(p_t,q_t)\} \in O(d \log d).$ We note that the latter requirement is aligned with the notion of ``budget violation'' as in \citet{LunghiCM26}, while it is orthogonal to the global budget balance condition \citep{BernasconiCCF24,GaucherBCCP25}, where it is required that $\sum_{t=1}^T \profit_t(p_t,q_t) \ge 0$. We remark that per-step budget balance is strictly stronger than all other notions.
    
    \paragraph{Steiner Potentials.}A crucial ingredient in our learning algorithms is provided by the Steiner Potentials \citep{LiuLS21}, inspired to the notion of Steiner Polynomials from convex geometry. Given a generic convex set $K$, the Steiner polynomial at scale $z\in \R_+$ is defined as the volume of the set $K + z\B$, where the sum is to be intended according to  Minkowski: $K + z\B = \{y + zw \mid y \in K, w \in \B\}$. 

    \paragraph{Confidence Regions.}All our algorithms maintain two confidence regions: $S_t$ for the seller and $B_t$ for the buyer, and update them according to the feedback received. These regions are initialized to $\B$, and are progressively updated according to hyperplanes orthogonal to the contexts. Formally, if at time $t$ the seller receives price $p_t$ and the algorithm observes $\ind{s_t \le p_t}$, then $S_{t+1}$ is:
    \[
    S_{t+1} = 
    \begin{cases}
    \{v \in S_{t}| \ip{v,x_t} \ge p_t\}  &\text{if the seller refuses $p_t$} \\
    \{v \in S_{t}| \ip{v,x_t} \le p_t\}  &\text{if the seller accepts $p_t$}
    \end{cases}
    \]
    \def\sellerbody{plot[smooth cycle] coordinates {(-1, 0.4) (0, 1.2) (1, 0.8) (0.8, -0.2) (-0.5, -0.3)}}
\def\buyerbody{plot[smooth cycle] coordinates {(-0.6, 0.3) (0, 0.8) (0.6, 0.5) (0.4, -0.1) (-0.3, -0.15)}}



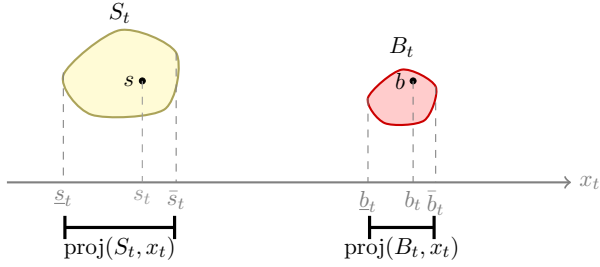
\begin{figure}[t!]
    \centering
    \begin{tikzpicture}[scale=0.75, every node/.style={scale=0.8}]
        \draw[->, gray, thick] (-5,0) -- (5,0) node[right] {$x_t$};

        \begin{scope}[shift={(-3, 1.5)}] 
            \draw[draw=yellow!60!black, fill=yellow!20, thick] \sellerbody;
            \node at (0, 1.5) {$S_t$}; 

            \coordinate (s_left) at (-1, 0.45);   
            \coordinate (s_right) at (1, 0.8);    
            
            \coordinate (s_yel) at (0.4, 0.3);
            
            \fill[draw=black, fill=black, line width=0.2pt] (s_yel) circle (1.5pt) node[left] {$s$};
        \end{scope}

        \draw[dashed, gray] (s_left) -- (s_left |- 0,0) node[below] {$\underline{s}_t$};
        \draw[dashed, gray] (s_right) -- (s_right |- 0,0) node[below] {$\bar s_t$};
        
        \draw[dashed, gray!80] (s_yel) -- (s_yel |- 0,0) node[below] {$s_t$};
        
        \draw[line width=1.2pt, |-|] (s_left |- 0,-0.8) -- (s_right |- 0,-0.8) node[midway, below] {$\its$};

        \begin{scope}[shift={(2, 1.2)}] 
            \draw[draw=red!80!black, fill=red!20, thick] \buyerbody;
            \node at (0, 1.2) {$B_t$}; 

            \coordinate (b_left) at (-0.6, 0.3);  
            \coordinate (b_right) at (0.6, 0.5);  
            
            \coordinate (b_red) at (0.2, 0.6);

            \fill[draw=black, fill=black, line width=0.2pt] (b_red) circle (1.5pt) node[left] {$b$};
        \end{scope}

        \draw[dashed, gray] (b_left) -- (b_left |- 0,0) node[below] {$\underline{b}_t$};
        \draw[dashed, gray] (b_right) -- (b_right |- 0,0) node[below] {$\bar b_t$};
        
        \draw[dashed, gray] (b_red) -- (b_red |- 0,0) node[below] {$b_t$};

        \draw[line width=1.2pt, |-|] (b_left |- 0,-0.8) -- (b_right |- 0,-0.8) node[midway, below] {$\itb$};

    \end{tikzpicture}
    \caption{Illustration of seller and buyer conference regions and projections.}
    \label{fig:visualization-intro} 
\end{figure}
    The buyer set $B_t$ is updated similarly, but with flipped signs. All our algorithms only update the confidence regions when there is no \emph{ambiguity} on $\ind{s_t \le p_t}$ and $\ind{q_t \le b_t}$. By construction, the confidence regions are convex and compact, therefore the sets $\its = \{\ip{v,x_t} \mid v \in S_t\}$ and $\itb=\{\ip{v,x_t} \mid v \in B_t\}$\footnote{The intervals $\its$ and $\itb$ are indeed the projections of $S_t$ and $B_t$ along the direction of $x_t$.} are closed and bounded intervals. Formally, let $\ls = \min\{\ip{v,x_t}| v \in S_t\}, \us = \max\{\ip{v,x_t}| v \in S_t\}$ (and analogously for the buyer), then $\its=[\ls,\us]$, and $\itb = [\lb,\ub]$ whose length we call width: The width of a convex set $K$ along direction $x$ is
    \(
        \width(K,x) = \max_{v_1 \in K} \ip{v_1,x} - \min_{v_2 \in K} \ip{v_2,x}.
    \)
    We refer to \Cref{fig:visualization-intro} for visualization.

\section{Context-free Bilateral Trade}\label{subsec:context-free}

We start by presenting the (simpler) \emph{context-free setting} (where each
$s_t = s$ and $b_t = b$ for some fixed but unknown values $s,b \in
[0,1]$ with $s \leq b$), in the one-bit feedback. While these algorithms 
don't require all the geometric machinery necessary in higher dimensions,
they illustrate some of the difficulties that arise in bilateral trade
that didn't exist one-sided pricing, and allow us to explain some building blocks used in subsequent algorithms. We first give an efficiency-maximizing
algorithm, and then extend it to a profit-maximizing algorithm.

\subsection{Efficiency maximization}
\label{sec:cont-free-effic-maxim}

For gain-from-trade maximization, we describe how to locate a price $p$ such that $s \leq p \leq b$, while suffering at most $O(1)$ regret. The main idea here is clean: the algorithm posts a fixed ``dyadic'' schedule of prices $\nicefrac 12, \nicefrac 14, \nicefrac 34, \nicefrac 18, \nicefrac 38, \nicefrac 58, \nicefrac 78, \hdots$, until a trade is observed. Formally, if
$t \in [2^i, 2^{i+1})$ and the learner has not observed a trade so
far, it posts $p_t = \nicefrac{(1+2(t-2^{i}))}{2^{i+1}}$. Once a
``good'' $p$ is found, the algorithm keeps posting it indefinitely,
and suffers no additional regret.

The main observation is that each price rejection shows that $(s,b)$
does not lie in some rectangle (indicated by the shaded regions in \Cref{fig:onedim-gft})---e.g., if price $\nicefrac 12$ is rejected, then
$(s,b) \notin [0,\nicefrac 12]\times[\nicefrac 12,1]$. Crucially, rejections also provide information of the length of $[s,b]$. If $\nicefrac 12$ is
rejected, then $b-s \leq \nicefrac12$; if both
$\nicefrac 14, \nicefrac 34$ are also rejected, we know that $b-s \leq \nicefrac 14$, etc. In general, if
the first $2^k-1$ prices are rejected, we conclude that
$b-s \leq \nicefrac 1{2^k}$, so that the total regret up to that point
is $(2^k -1) \cdot (b-s) \in O(1)$. 

\paragraph{Random Prices.} There is an intriguing alternate way of
simulating the dyadic search algorithm via randomization: post random
prices drawn uniformly and independently in $[0,1]$ until a trade is
observed, then play that price for the rest of the time horizon. The
number of price rejections is a geometric random variable with
parameter ${(b-s)}$, so that the expected number of rejections is
$\nicefrac{1}{(b-s)}$, for an overall expected regret of $1$. In our
contextual algorithms, we use both a ``dyadic-like'' structured
approach, as well as the randomized one.

\begin{figure}[t!]
\centering

\begin{tikzpicture}[scale=0.7]

\draw (0,0) -- (4,4);
\draw[fill=red!20!white] (0,2) rectangle (2,4);
  \node [shape=circle, fill=black,inner sep=1.5pt,label=right:$\nicefrac 12$] (2,2) at (2,2) {};
\draw[line width=1.5pt] (0,0) rectangle (4,4);

\begin{scope}[xshift=5cm]
\draw (0,0) -- (4,4);
\draw[fill=red!20!white] (0,2) rectangle (2,4);
\draw[fill=red!20!white] (0,1) rectangle (1,2);
\draw[fill=red!20!white] (2,3) rectangle (3,4);
  \node [shape=circle, fill=black,inner sep=1.5pt,label=right:$\nicefrac 14$] (1,1) at (1,1) {};
  \node [shape=circle, fill=black,inner sep=1.5pt,label=right:$\nicefrac 34$] (3,3) at (3,3) {};
\draw[line width=1.5pt] (0,0) rectangle (4,4);
\end{scope}

\begin{scope}[xshift=10cm]
\draw (0,0) -- (4,4);
\draw[fill=red!20!white] (0,2) rectangle (2,4);
\draw[fill=red!20!white] (0,1) rectangle (1,2);
\draw[fill=red!20!white] (2,3) rectangle (3,4);
\draw[fill=red!20!white] (.5,.5) rectangle (0,1);
\draw[fill=red!20!white] (1.5,1.5) rectangle (1,2);
\draw[fill=red!20!white] (2.5,2.5) rectangle (2,3);
\draw[fill=red!20!white] (3.5,3.5) rectangle (3,4);
  \node [shape=circle, fill=black,inner sep=1.5pt,label=right:$\nicefrac 18$] (.5,.5) at (.5,.5) {};
  \node [shape=circle, fill=black,inner sep=1.5pt,label=right:$\nicefrac 38$] (1.5,1.5) at (1.5,1.5) {};
    \node [shape=circle, fill=black,inner sep=1.5pt,label=right:$\nicefrac 58$] (2.5,2.5) at (2.5,2.5) {};
  \node [shape=circle, fill=black,inner sep=1.5pt,label=below:$\nicefrac 78$] (3.5,.35) at (3.5,3.5) {};
\draw[line width=1.5pt] (0,0) rectangle (4,4);

\end{scope}
\end{tikzpicture}
\caption{Visualization of the context-free dyadic search. $x$-axis corresponds to seller, $y$-axis to buyer.}
\label{fig:onedim-gft}
\end{figure}
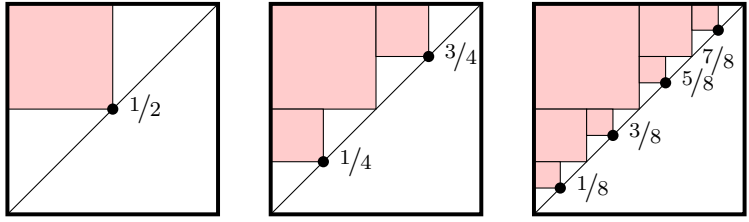
    
\subsection{Profit maximization}
\label{sec:cont-free-profit-maximization}

   For profit-maximization, we give an algorithm that suffers $O(\log \log T)$ regret. In this case, staying with the first accepted price is not enough---we also want the prices $(p_t,q_t)$ to quickly approach $(s,b)$.
    The algorithm runs in two steps. First, it performs the dyadic
    search used for efficiency maximization, suffering at most $O(1)$
    regret. Once it observes a trade for some price $p$, then we have
    identified a ``nice'' square $Q = [p-\eps, p] \times [p, p+\eps]$
    with $(s,b) \in Q$---specifically, if $p_t =
    \nicefrac{(1+2(t-2^{i}))}{2^{i+1}}$ is accepted then $(s,b) \in Q
    = [p-2^{-i-1}, p] \times [p, p+2^{-i-1}]$.

    At this point, the algorithm moves to its second step, and recursively zooms to explore the square $Q$ and locate $(s,b)$. This zooming is recursively performed via a two-dimensional version of the \emph{quadratic search} algorithm, inspired by \citet{KleinbergL03}. 
    Given a generic square $Q^i$ of side $\eps_i$, whose bottom-right corner is $(s^i,b^i)$, the algorithm casts a uniform $\eps_{i+1}=\eps_i^2$ grid inside it and starts offering prices $s^i$, $s^i-\eps_{i+1}, s^i-2\eps_{i+1}, \dots$ to the seller, while keeping $b^i$ to the buyer. As soon as the seller refuses, we can identify a vertical strip of width $\eps_{i+1}$ that contains $(s,b)$ (see \Cref{fig:onedim-profit-a}). We can perform the same operation for the buyer, while keeping fixed the seller price to $s^i$, identifying now a horizontal strip of height $\eps_{i+1}$. Now $(s,b)$ belongs to the intersection of these two stripes, i.e., a square of side $\eps_{i+1}$, into which we recuse (see \Cref{fig:onedim-profit-b}). The recursion can stop as soon as the side of the square $Q^i$ drops below $\nicefrac 1T$ (as this is enough to obtain the target profit, with a constant overall error), i.e., after $O(\log \log T)$ phases.

    For the regret analysis, observe that the algorithm issues $O(\nicefrac{1}{\eps_i})$ prices in each phase. When the trade happens, the regret suffered is $O(\eps_i)$ since the side of $Q^i$ (and thus the error in estimating the target profit) is of that order. In the two steps where the trade does not happen, the regret is constant. Overall, the regret in the generic phase $i$ is constant. Since there are at most $O(\log \log T)$ such phases, this implies that the overall regret is $O(\log \log T)$.

    \begin{figure}[t!]
\centering
\scalebox{.8}{
\begin{subfigure}[b]{0.4\textwidth}
    \centering
    \begin{tikzpicture}[scale=1.2]
    \draw[line width=1.5pt] (0,0) rectangle (4,4);

    \foreach \i in {0,...,4} {
        \foreach \j in {0,...,4} {
            \node [shape=circle, fill=black,inner sep=1.5pt] (\i,\j) at (\i,\j) {};
        }
    }

    \node at (2, -0.4) {$Q^{i}$}; 

    \draw[|<->|] (0,-.4) -- node[below=2pt] {$\varepsilon_{i+1}$} (1,-.4); 
    \draw[|<->|] (-.4,0) -- node[left=2pt] {$\varepsilon_{i}$} (-.4,4); 

    \draw[<-, line width=1.2pt, green!70!black] (0.2, .5) -- (3.8, .5);

    \draw[line width=1.5pt, green!70!black, dashed] (2, 0) -- (2, 4); 

    \node at (4, -0.4) {$(s^{i},b^{i})$};

    \end{tikzpicture}
    \caption{Seller's price discovery in phase $i$}
    \label{fig:onedim-profit-a}
\end{subfigure}
\hfill
\begin{subfigure}[b]{0.4\textwidth}
    \centering
    \begin{tikzpicture}[scale=1.2]
    \draw[line width=1.5pt] (0,0) rectangle (4,4);

    \draw[rounded corners, fill=purple!20!white, line width=0pt] (2, 1) rectangle (3, 2);

    \foreach \i in {0,...,4} {
        \foreach \j in {0,...,4} {
            \node [shape=circle, fill=black,inner sep=1.5pt] (\i,\j) at (\i,\j) {};
        }
    }

    \node at (2, -0.4) {$Q^{i}$}; 

    \draw[|<->|] (0,-.4) -- node[below=2pt] {$\varepsilon_{i+1}$} (1,-.4); 
    \draw[|<->|] (-.4,0) -- node[left=2pt] {$\varepsilon_{i}$} (-.4,4); 

    \draw[<-, line width=1.2pt, blue] (3.5, 3.8) -- (3.5, 0.2);

    \draw[line width=1.5pt, green!70!black, dashed] (2, 0) -- (2, 4); 

    \draw[line width=1.5pt, blue, dashed] (0, 2) -- (4, 2); 

    \node[below right, xshift=0.2cm, yshift=0.2cm, color=black] at (1.85,1.5) {$Q^{i+1}$}; 

    \node at (4, -0.4) {$(s^{i},b^{i})$};

    \end{tikzpicture}
    \caption{Buyer's price discovery in phase $i$}
    \label{fig:onedim-profit-b}
\end{subfigure}





}
\caption{Illustration of the quadratic search step: In phase $i$, we operate within a square $Q^i$ of side $\varepsilon_i$. (a) We fix the buyer's price to $b^i$, then test decreasing seller prices starting from $s^i$. The vertical dashed line indicates the first refused seller price. (b) We then fix the seller's price to $s^i$, then test increasing buyer prices starting from $b^i$. The horizontal dashed line indicates the first refused buyer's price. These two tests identify a new square $Q^{i+1}$ of side $\eps_{i+1}$, where the algorithm is called recursively.}

\label{fig:onedim-profit}
\end{figure}
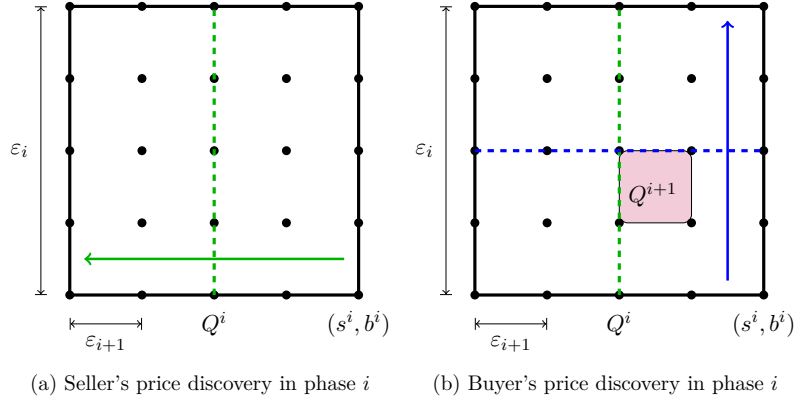
\section{The Contextual Case: Efficiency Maximization}
\label{sec:efficiency}

We move on to the contextual case, and present the three
algorithms for efficiency maximization; recall that the efficiency is measured in terms of the cumulative gain from trade. In \Cref{subsec:two-bit-gft}, we present the algorithm for two-bit feedback and prove its optimality, then in \Cref{subsec:one-to-two-bit-gft} we illustrate how to adapt it to one-bit feedback with a minimal budget-balance violation. Finally, \Cref{subsec:one-bit-gft} is devoted to presenting the randomized one-bit algorithm that enforces per-round budget balance.

    \subsection{Efficiency Maximization with Two-Bit Feedback}
    \label{subsec:two-bit-gft}

        The algorithm\footnote{We refer to \Cref{alg:2bit-gft} in \Cref{app:pseudocodes-efficiency} for the pseudocode.} for efficiency maximization with two-bit feedback maintains the two confidence regions $S_t$ and $B_t$ and uses an auxiliary family of potentials $\vol(S_t + z_i\B)$ and $\vol(B_t + z_i\B)$, for $z_i = \nicefrac{2^{-i}}{(8d)}$ and integers $i$. At each time step, the algorithm decides what to do according to the widths ($\width(S_t,x_t)$ and $\width(B_t,x_t)$) of the confidence regions along $x_t$. It considers three cases: 
        \begin{itemize}
            \item[(i)] if the intervals $\its$ and $\itb$ are disjoint, then it posts the seller \emph{safe price} $p_t = \us$ to both agents (see \Cref{fig:visualization-gft-a})
            \item[(ii)] if $\width(S_t,x_t) \ge \width(B_t,x_t)$, then it posts to both agents the price $p_t$ which solves
        \begin{equation}
        \label{eq:balanced-price-seller}
            \vol(\{v \in S_t + z_{i_t}\B \mid \ip{v,x_t} \le p_t\}) = \tfrac{1}{2} \vol(S_t + z_{i_t} \B),
        \end{equation}
        where $i_t$ is the largest index such that $\width(x_t,S_t) \le 2^{-i}$ (see \Cref{fig:visualization-gft-b}); 
            \item[(iii)] otherwise, it posts to both agents the price $p_t$ which solves
        \begin{equation}
        \label{eq:balanced-price-buyer}
            \vol(\{v \in B_t + z_{i_t}\B \mid \ip{v,x_t} \ge p_t\}) = \tfrac{1}{2} \vol(B_t + z_{i_t} \B),
        \end{equation}
        where $i_t$ is the largest index such that $\width(x_t,B_t) \le 2^{-i}$.
        \end{itemize} The prices defined in \Cref{eq:balanced-price-seller,eq:balanced-price-buyer} are called balanced, as they reduce the corresponding potential by a constant factor, \emph{regardless of the feedback}. If at time $t$ the algorithm posts a balanced price, $i_t$ defined above is called the \emph{index} at time step $t$. 
        Before proving the regret guarantees of our algorithm, we state a lemma formalizing the decrease in potential due to balanced prices. It is an adaptation of Lemma 2.1 in \citet{LiuLS21}, and is re-proved for completeness in \Cref{app:efficiency}.
\def\sellerbody{plot[smooth cycle] coordinates {(-1, 0.4) (0, 1.2) (1, 0.8) (0.8, -0.2) (-0.5, -0.3)}}
\def\buyerbody{plot[smooth cycle] coordinates {(-0.6, 0.3) (0, 0.8) (0.6, 0.5) (0.4, -0.1) (-0.3, -0.15)}}

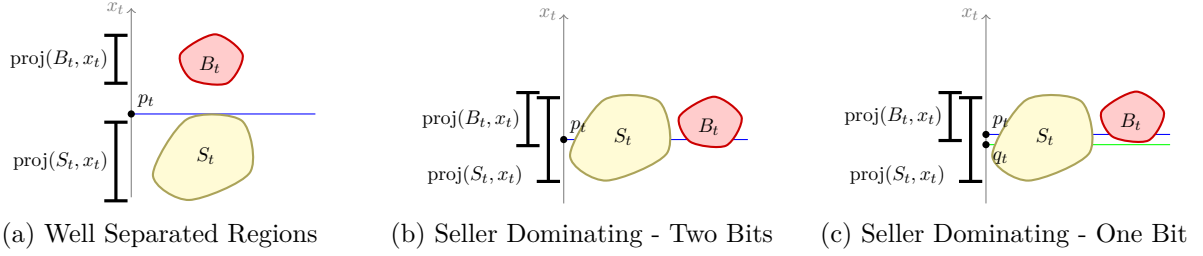
\begin{figure}[t!]
    \centering
    
    
    
    
    \begin{subfigure}[b]{0.32\textwidth}
    \centering
        \begin{tikzpicture}[scale=0.68, every node/.style={scale=0.68}]
            \draw[->, gray] (0,-2.5) -- (0,1.2) node[left] {$x_t$};
            \draw[-, blue] (0,-0.87) -- (3.6,-0.87);
            
            \begin{scope}[shift={(1.5, -2.1)}, rotate=20] 
                \draw[draw=yellow!60!black, fill=yellow!20, thick] \sellerbody;
                \node (s) at (0.1, 0.4) {$S_t$};
            \end{scope}
    
            \begin{scope}[shift={(1.5, -0.1)}, rotate=-10] 
                \draw[draw=red!80!black, fill=red!20, thick] \buyerbody;
                \node (b) at (0, 0.2) {$B_t$};
            \end{scope}
    
            \draw[line width=1.2pt, |-|] (-0.3,-1) -- (-0.3,-2.6) node[midway, left=0pt, color=black] {$\its$}; 
    
            \draw[line width=1.2pt, |-|] (-0.3,-0.3) -- (-0.3,0.7) node[midway, left=0pt, color=black] {$\itb$}; 
    
            \fill (0,-0.87) circle (2pt) node[above right] {$p_t$};
        \end{tikzpicture}
        \caption{Well Separated Regions
        }
        \label{fig:visualization-gft-a}
    \end{subfigure}
    \hfill
    \begin{subfigure}[b]{0.32\textwidth}
    \centering
        \begin{tikzpicture}[scale=0.68, every node/.style={scale=0.68}]
            \draw[->, gray] (0,-2.5) -- (0,1.2) node[left] {$x_t$};
            \draw[-, blue] (0,-1.25) -- (3.6,-1.25);
            
            \begin{scope}[shift={(1.2, -1.6)}, rotate=20] 
                \draw[draw=yellow!60!black, fill=yellow!20, thick] \sellerbody;
                \node (s_label_in_body) at (0.1, 0.4) {$S_t$};
            \end{scope}
    
            \begin{scope}[shift={(2.8, -1.2)}, rotate=-10] 
                \draw[draw=red!80!black, fill=red!20, thick] \buyerbody;
                \node (b_label_in_body) at (0, 0.2) {$B_t$};
            \end{scope}
    
            \draw[line width=1.2pt, |-|] (-0.3,-2.1) -- (-0.3,-0.4) node[midway, below left=10pt, color=black] {$\its$}; 
    
            \draw[line width=1.2pt, |-|] (-0.7,-1.4) -- (-0.7,-0.3) node[midway, left=0pt, color=black] {$\itb$}; 
    
            \fill (0,-1.25) circle (2pt) node[above right] {$p_t$};
        \end{tikzpicture}
        \caption{Seller Dominating - Two Bits}
        \label{fig:visualization-gft-b}
    \end{subfigure}
    \hfill
    \begin{subfigure}[b]{0.32\textwidth}
    \centering
        \begin{tikzpicture}[scale=0.68, every node/.style={scale=0.68}]
            \draw[->, gray] (0,-2.5) -- (0,1.2) node[left] {$x_t$};
            \draw[-, blue] (0,-1.15) -- (3.6,-1.15);
            \draw[-, green] (0,-1.35) -- (3.6,-1.35);
            
            \begin{scope}[shift={(1.2, -1.6)}, rotate=20] 
                \draw[draw=yellow!60!black, fill=yellow!20, thick] \sellerbody;
                \node (s_label_in_body) at (0.1, 0.4) {$S_t$};
            \end{scope}
    
            \begin{scope}[shift={(2.8, -1.1)}, rotate=-10] 
                \draw[draw=red!80!black, fill=red!20, thick] \buyerbody;
                \node (b_label_in_body) at (0, 0.2) {$B_t$};
            \end{scope}
    
            \draw[line width=1.2pt, |-|] (-0.3,-2.1) -- (-0.3,-0.4) node[midway, below left=10pt, color=black] {$\its$}; 
    
            \draw[line width=1.2pt, |-|] (-0.7,-1.3) -- (-0.7,-0.3) node[midway, left=0pt, color=black] {$\itb$}; 
    
            \fill (0,-1.15) circle (2pt) node[above right] {$p_t$};
            \fill (0,-1.35) circle (2pt) node[below right] {$q_t$};
        \end{tikzpicture}
        \caption{Seller Dominating - One Bit}
        \label{fig:visualization-gft-c}
    \end{subfigure}

    \caption{Visualization of the efficiency maximization algorithms.}
    \label{fig:visualization-gft} 
\end{figure}
        \begin{restatable}[Balanced Prices]{lemma}{lembalanced}
        \label{lem:balanced}
            If the algorithm posts a balanced price for the seller, then 
            \[
                \vol(S_{t+1} + z_{i_t}\B) \le \tfrac{3}{4} \vol(S_{t} + z_{i_t}\B).
            \]
            If instead it posts a balanced price for the buyer, then  
            \[
                \vol(B_{t+1} + z_{i_t}\B) \le \tfrac{3}{4} \vol(B_{t} + z_{i_t}\B).
            \]
        \end{restatable}
        
    \begin{theorem}
    \label{thm:2-bit-gft}
        Consider the contextual bilateral trade problem in the two-bit feedback model. There exists an efficiency-maximizing algorithm that achieves regret $O(d \log d)$ and enforces per-round budget balance. 
    \end{theorem}
    \begin{proof}
        The algorithm always posts the same price to both agents, so it naturally respects per-round budget balance. Moreover, it never incurs a \emph{negative} gain from trade as a trade may occur only if $s_t \le p_t \le b_t$ (and thus $b_t \ge s_t$). We can then discard from the analysis all time steps in which $s_t > b_t$, as they do not induce positive gain from trade for the algorithm or the benchmark.

        Consider now the ``well-separated'' time steps, where $\its \cap \itb = \emptyset$, they always result in a trade as both agents accept and the algorithm suffers zero regret:
        \[
            s_t = \ip{x_t,s} \le \us = p_t < \lb \le \ip{b,x_t} = b_t.
        \]
        We are left with analyzing the regret incurred by the algorithm in the seller/buyer dominating cases. The crucial observation is that every time we post a balanced price corresponding to some index $i_t$, either $\vol(S_t + z_{i_t}\B)$ or $\vol(B_t + z_{i_t}\B)$ decreases by a constant factor (\Cref{lem:balanced}). Both confidence regions are never empty, so that all potentials are always at least $\vol(z_{i}\B) = z_{i}^d \vol(\B)$, while they are at most $\vol((z_{i}+1)\B) \le \vol(3\B) = 3^d \vol(\B)$ by definition. Denoting by $n^S_i$ the number of times that the index $i$ is selected in a seller-dominating balanced case, we have the following inequality:
        \[
            \left(\tfrac 34\right)^{n^S_i} 3^d \vol(\B) \ge z_{i}^d \vol(\B) \quad \implies \quad n_i^S \le \frac{d \log \nicefrac{3}{z_i}}{\log \nicefrac{4}{3}} \le 3 d \log \nicefrac{3}{z_i}.
        \]
        The same bound holds for the number of times that index $i$ is selected in a buyer-dominating balanced case. All in all, the algorithm does not pick any index $i$ more than $ 6 d \log(\nicefrac3{z_{i}})$ times.

        Consider now the regret suffered in a generic step when index $i$ is chosen, either because of the buyer or of the seller. Since the two intervals have a non-empty intersection, and the longest one has length at most $2^{-i}$, we can conclude that the gain from trade of the benchmark is at most $2 \cdot 2^{-i}$, which provides an upper bound on the regret suffered. Combining these ingredients, we can conclude that the overall regret is at most 
        \begin{equation}
            \label{eq:bound_balanced}
            \sum_{i=-1}^{\infty} 12 \cdot 2^{-i} d \log(\nicefrac3{z_i}) \le 12d \sum_{i=-1}^{\infty}\frac{1}{2^i}(\log(24d) + i \log(2)) \in O(d \log d),
        \end{equation}
        where the first inequality follows from the definition of $z_i$ as $\nicefrac{2^{-i}}{(8d)}$.
        \end{proof}
    We complement this positive result with a lower bound, showing that the above result is essentially tight, up to a $\log d$ factor. The lower bound construction is based on the simple observation that any learning algorithm cannot avoid suffering constant regret along one of the $d$ directions, and its proof is deferred to \Cref{app:efficiency}.
 
    \begin{restatable}{proposition}{lowergft}
    \label{thm:2-bit-gft-lower}
         Consider the contextual bilateral trade problem in the two-bit feedback model.  Any efficiency-maximizing algorithm that enforces per-round budget balance suffers regret $\Omega(d)$.
    \end{restatable}

    \subsection{Efficiency Maximization with One-Bit Feedback}
    \label{subsec:one-to-two-bit-gft}

        With a simple tweak, we can apply the previous algorithm (for efficiency maximization with two-bit feedback) to the one-bit feedback setting: whenever there may be ambiguity in the feedback observed by the learner, we post a \emph{safe} price to one of the two agents, to disambiguate the feedback.

        In the well-separated case, there is no ambiguity, as both agents always accept the posted price (or $\ls \ge \ub$ and the trade clearly cannot happen, in which case we do not update anything). For the other two cases, we post a safe price to the dominated agent: if the width of $S_t$ is larger than that of $B_t$, then we post the seller-balanced price to the seller, and a safe price to the buyer (i.e.,  the smallest value $\lb$ in $\itb$ see \Cref{fig:visualization-gft-c}), similarly, if the buyer's width dominates the seller's, then the algorithm post the buyer-balanced price to the buyer, and a safe price to the seller (i.e., the largest value $\us$ in $\its$). This simple modification is enough to obtain the following result.

    \begin{theorem}
    \label{thm:gft-safe}
        Consider the contextual bilateral trade problem in the one-bit feedback model. There exists an efficiency-maximizing algorithm that achieves regret $O(d \log d)$ and whose negative profit is at most $O(d \log d)$. 
    \end{theorem}
    \begin{proof}
        First, note that the algorithm is well-defined, as it always posts a safe price to one of the two agents so that one-bit feedback is enough to update both confidence regions correctly. Moreover, the potential argument is not affected, as it depends on the balanced prices the algorithm still posts. This means that the bound proved in \Cref{thm:2-bit-gft} on the number of times that a given index $i$ is chosen holds, and this is at most $O(d \log(\nicefrac 1{z_i}))$. 
        We want that (i) the gain from trade of the algorithm does not deteriorate too much and (ii) the overall negative profit is at most $O(d \log d).$ 
        
        We start considering the gain from trade. The only difference with the analysis of the two-bit case lies in the fact that the algorithm now may incur \emph{negative} gain from trade when posting prices that are not budget balanced. However, such negative regret is at most of order $2^{-i}$ (for the index $i$ in that step) because both prices lie in the interval with the largest length. All in all, the extra-regret incurred by the one-bit algorithm is bounded as in \Cref{eq:bound_balanced}.
        Similarly, the algorithm may suffer negative profit only when it posts two non-budget-balanced prices. This happens only when $\its$ and $\itb$ have a non-empty intersection, with the two prices being at most $2^{-i}$ far apart, where $i$ is the index selected in the given time step. The $O(d \log d)$ bound on the negative profit follows by the same analysis as in \Cref{eq:bound_balanced}.
    \end{proof}

    \subsection{Efficiency Maximization with One-Bit Feedback and Per-Round Budget Balance}
    \label{subsec:one-bit-gft}

    We move to efficiency-maximization under one-bit feedback and per-round budget balance. Our algorithm posts the same price $p_t$ to both agents and updates confidence regions $S_t$ and $B_t$, only when there is no ambiguity in the (one-bit) feedback. At each time step $t$, the algorithm computes the intervals $\its$ and $\itb$ and considers the following three cases:
        \begin{itemize}
        \item \emph{Well-separated case:} if the two intervals are disjoint, then it posts a ``safe'' price for the seller, i.e., $\us$, the largest value in $\its$;
        \item \emph{Weak overlap case:} let $p_t$ be the middle-point of the larger interval, if it is accepted for sure by the other agent, then the algorithm posts it, \emph{middle price};
        \item \emph{Strong overlap case:} when neither of the two previous cases apply, the algorithm posts a price uniformly at random in $\its \cup \itb$, \emph{uniform price}.
        \end{itemize}
A pseudocode description of the algorithm is available in \Cref{alg:1bit-gft}, in \Cref{app:pseudocodes-efficiency}.

Note that, unlike the other algorithms for efficiency maximization, \Cref{alg:1bit-gft} is defined without explicitly relying on Steiner Potentials; nonetheless, they remain a crucial tool in the analysis.
In the rest of this section, we argue that this one-bit algorithm suffers a regret that is independent of the time horizon. For this purpose, we present a geometric result, \Cref{lem:partition}, whose proof is deferred to  \Cref{app:efficiency}. 
\begin{restatable}{lemma}{lempartition}
        \label{lem:partition}
            Consider any convex set $K\subseteq \R^d$, direction $x \in \B$, and hyperplane $h$ orthogonal to $x$.
            Let $H$ be any half-space induced by $h$, such that $\proj(x, H)$ contains an $\alpha \in [0,1]$ fraction of $\proj(x, K)$. 
            For any $z \le \alpha \width(K,x)$, 
            we have the following inequality:
            \[
                \vol((K \setminus H) + z\B)\leq \left(1-\left(\frac{\alpha}{1+2\alpha}\right)^d\right) \cdot \vol(K + z\B).     
            \]
        \end{restatable}



We now present two direct applications of \Cref{lem:partition} in  \Cref{lem:unbalanced} and \Cref{lem:balanced-1bit}.
\begin{lemma}[Weak overlap case]
        \label{lem:unbalanced}
            In the weak overlap case, if the algorithm posts the middle point of the seller, then 
            \[
                \vol(S_{t+1} + z \B) \le (1-4^{-d})\vol(S_{t} + z \B), \quad \forall z \text{ such that } z \le \nicefrac 12 \width(S_t,x_t).
                \]
            If the algorithm posts the middle point of the buyer, then:
            \[
                \vol(B_{t+1} + z \B) \le (1-4^{-d})\vol(B_{t} + z \B), \quad \forall z \text{ such that } z \le \nicefrac 12 \width(B_t,x_t).
            \]
        \end{lemma}
        \begin{proof}
            Without loss of generality, we assume that the seller has the larger width, so that the algorithm posts the middle point of $\its$ as price, and that such price is smaller than $\lb$. 
            By design, the buyer is sure to accept the middle price of the seller, thus the feedback received by the algorithm allows to reconstruct whether the seller accepts or not. We can then apply \Cref{lem:partition} with $\alpha = \nicefrac 12$ to conclude the proof.
        \end{proof}

\begin{lemma}[Strong overlap case]
    \label{lem:balanced-1bit}
In the strong overlap case, if both agents accept price $p_t$, then at least one of the following two inequalities is verified: 
        \[
        \begin{cases}
            \vol(S_{t+1} + z \B) \le (1-6^{-d})\vol(S_{t} + z \B), \quad &\forall z \text{ such that } z \le \nicefrac{1}{4}\width(S_t,x_t)    \\
            \vol(B_{t+1} + z \B) \le (1-6^{-d})\vol(B_{t} + z \B), \quad &\forall z \text{ such that } z \le \nicefrac{1}{4}\width(B_t,x_t)
        \end{cases}.
        \]
    \end{lemma}
    \begin{proof}
        Assume, towards contradiction, that the statement is not true, then there exists a price $p_t$ that (i) results in a trade, (ii) belongs to the leftmost quarter of the $[\lb,\ub]$ segment and (iii) belongs to the rightmost quarter of the $[\ls,\us]$ segment. Note, the properties (ii) and (iii) are required, otherwise we could apply \Cref{lem:partition} with $\alpha = \nicefrac 14$.

        Consider the case that the seller's width is larger than the buyer's one, 
        then the existence of $p$ satisfying (i) and (iii) implies that the middle point of $[\ls,\us]$ is smaller than $\lb$. This is because $\nicefrac{(\us-\ls)}{2} = \us - \nicefrac{(\us+\ls)}{2} = (\us - p) + (p-\lb) + (\lb-\nicefrac{(\us+\ls)}{2})$ gives $(\lb-\nicefrac{(\us+\ls)}{2})\geq 0$ using $(\us-p)\leq \nicefrac{(\us-\ls)}{4}$ and $(p-\lb)\leq \nicefrac{(\us-\ls)}{4}$.  
But this is a contradiction, because we would then be in the weak overlap case. The case where the buyer's width dominates the seller's is analogous. 
    \end{proof}

We are now ready to prove the main theorem of this section. 
        \begin{theorem}
        \label{thm:one-bit-gft}
            Consider the contextual bilateral trade problem in the one-bit feedback model. There exists an efficiency-maximizing algorithm that achieves regret $O(d6^d)$ and enforces per-round budget balance. 
        \end{theorem}


\begin{proof} It is clear that Algorithm \ref{alg:1bit-gft} enforces per-round budget balance, as it always posts the same price to the buyer and to the seller, and is well defined, as it only updates the confidence regions when there is no ambiguity. We then tackle bounding the regret rate. Fix any problem instance $\cI$, we introduce two auxiliary functions: a potential $P(t)$, which is a function of the state of the algorithm at time $t$ and defined formally later, and the cumulative regret $R(t)$ suffered against $\cI$ suffered by the algorithm up to time $t$, excluded. We omit the dependence on the instance $\cI$, but it is clear that our goal is to bound $\E{R(T+1)}$ (with a slight abuse of notation).

The main ingredient in our analysis is the following inequality, which we shall show at all times $t\geq 1$:
\begin{equation}\label{eq:potential-eq}
    \mathbb{E}_t[R(t+1)+P(t+1)]\leq R(t)+P(t),
\end{equation} 
where $\mathbb{E}_t[\cdot]$ is the expectation conditioned on values taken by $p_1,\dots,p_{t-1}$ (so the randomness in the left-hand side of \Cref{eq:potential-eq} is on $p_t$). A direct induction on \Cref{eq:potential-eq} then yields the regret bound at all times $t$,
\[
\E{R(t)} \leq P(0).
\]
In formal terms, Equation \eqref{eq:potential-eq} says that $P(t) + R(t)$ is a super-martingale relative to the filtration $(\mathcal{F}_t)_{t\in [T+1]}$, where $\mathcal{F}_t$ is the sigma-algebra used used to generate $(p_1,\dots,p_{t-1}$). In this terminology, $\mathbb{E}_t[R(t+1)+P(t+1)]$, also denoted $ \mathbb{E}[R(t+1)+P(t+1)|\mathcal{F}_t]$, is an $\mathcal{F}_t$-measurable random variable, just like $R(t)+P(t)$. Similarly, we denote with $\mathbb{P}_t$ the probability conditioned on $\cF_t$. Thus, the inequality in Equation \eqref{eq:potential-eq} is an inequality between $\mathcal{F}_t$-measurable random variables.
    We introduce the potential $P$ that we use in this proof. Let $z_i = 2^{-i}$, we have 
\begin{equation}\label{eq:potential-def}
    P(t) =  64 \cdot 6^d\sum_{i \geq 0} z_i\left[\log \left(\frac{\vol(B_t + z_{i} \B)}{z_i^d\text{Vol}(\B)}\right) + \log \left(\frac{\vol(S_t + z_{i} \B)}{z_i^d\text{Vol}(\B)}\right)\right].
\end{equation}
The non-negativity of such potential immediately follows from the fact that, at all times, $B_t\neq \emptyset$ and $S_t\neq \emptyset$. Also, using that $\forall i: \vol(\B+z_i\B) \le \vol(2\B) = 2^d \vol(\B)$, we get,
\[
P(0) \le 128 \cdot 6^d \sum_{i\geq 0}z_i \log\left(\frac{2^d}{z_i^d}\right) \in O(d6^d),
\]
which corresponds to the desired regret guarantee.

Thus, to complete the proof, it suffices to show that Equation \eqref{eq:potential-eq} is true at all times for the potential introduced in \Cref{eq:potential-def}. To do so, it is useful to notice that all the terms in the definition of $P(t)$ are non-increasing, so that it suffices to show that the decrease of one of these terms outweighs the increase in the regret.  

\paragraph{Well-separated case.} In the well-separated case, the algorithm suffers no regret, i.e., $R(t+1) = R(t)$. Equation \eqref{eq:potential-eq} thus follows immediately from the monotonicity of $P(t)$. 

\paragraph{Weak overlap case.} In the  weak overlap case, notice that $R(t+1)-R(t)\leq 2w$, where $w = \max\{\width(S_t,x_t), \width(B_t,x_t)\}$.  
Suppose, without loss of generality, that $\width(S_t,x_t)\geq \width(B_t,x_t)$, so that $p_t$ is the middle point of $\its$ and is always accepted by the buyer. In this case, denoting by $i$ the smallest integer such that $z_i\leq \nicefrac{1}{2}\width(S_t,x_t)$, and applying Lemma \ref{lem:unbalanced}, we have $\vol(S_{t+1} + z_i \B) \le (1-4^{-d})\vol(S_{t} + z_i \B)$. Therefore, 
\begin{align*}
    P(t+1)-P(t)&\leq 64 \cdot 6^d z_i \log(1-4^{-d})\\
    &\leq -2w,
\end{align*}
where we used $\forall x\in (0,1]: \log(1-x)\leq -x$, and $4z_i\geq w$, by definition of $i$. 

\paragraph{Strong overlap Case.} In the strong overlap case, the increase in the regret satisfies $R(t+1)-R(t)\leq \ip{b-s,x_t}^+$, and the probability that the price selected uniformly at random in $\its \cup \itb$ is accepted satisfies\footnote{With $\mathbb{P}_t$ we denote the conditional probability with respect to the filtration $\cF_t$.} 
\[
\mathbb{P}_t[p_{t}~\text{accepted}]\geq \frac{\ip{b-s,x_t}^+}{2w}.
\]
Thus, in order to prove Equation \eqref{eq:potential-eq}, in light of the observation that for the non-increasing potential,
\[
\mathbb{E}_t[P(t+1)-P(t)] \leq  \mathbb{E}_t[P(t+1)-P(t)|p_{t}~\text{accepted}]\cdot\mathbb{P}_t[p_{t}~\text{accepted}],
\]
it suffices to show the following bound on the conditional expectation
\[
\mathbb{E}_t[P(t+1)-P(t) |p_{t}~\text{accepted}]\leq -2w.
\]
We will show the stronger claim that $P(t+1)-P(t)\leq -2w$ in all cases where $p_t$ is accepted. For this, we consider two sub-cases. 

\paragraph{Sub-case 1: $\width(B_t,x_t)$ and $\width(S_t,x_t)$ are within a factor $4$.} Denote by $i$ the smallest integer such that $z_i\leq \nicefrac{1}{16}\cdot w$ so that $z_i\leq \nicefrac{1}{4}\min\{\width(B_t,x_t),\width(S_t,x_t)\}$. By Lemma \ref{lem:balanced-1bit}, if the price is accepted, then at least one of the following two inequalities is verified: 
        \[
        \begin{cases}
            \vol(S_{t+1} + z \B) \le (1-6^{-d})\vol(S_{t} + z \B), \quad &\forall z \text{ such that } z \le \nicefrac{1}{4}\width(S_t,x_t)    \\
            \vol(B_{t+1} + z \B) \le (1-6^{-d})\vol(B_{t} + z \B), \quad &\forall z \text{ such that } z \le \nicefrac{1}{4}\width(B_t,x_t)
        \end{cases}
        \]
Therefore, the quantity $Q(t)$ defined by
\[
Q(t) = 64\cdot 6^d\cdot z_i\left(\log \left(\frac{\vol(B_t + z_{i} \B)}{z_i^d\text{Vol}(\B)}\right) + \log \left(\frac{\vol(S_t + z_{i} \B)}{z_i^d\text{Vol}(\B)}\right)\right),
\]
satisfies,
\[
Q(t+1)-Q(t)\leq 64\cdot 6^d\cdot z_i \log(1-6^{-d}) \leq -2w,
\]
where we used $\forall x\in (0,1]: \log(1-x)\leq -x$, and $32 \cdot z_i \geq w$.
Observing that $P(t+1)-P(t)\leq Q(t+1)-Q(t)$, we have completed the analysis of this sub-case. 

\paragraph{Sub-case 2: $\width(B_t,x_t)\geq 4\width(S_t,x_t)$ or $\width(S_t,x_t)\geq 4\width(B_t,x_t)$.} We assume without loss of generality that $\width(B_t,x_t)\geq 4\width(S_t,x_t)$ and we treat the other case analogously. Notice that an accepted price must belong to the last three-quarters of $B_t$, so that, taking $i$ the smallest integer such that   $z_i\leq \nicefrac{1}{4}w$, and applying Lemma \ref{lem:partition} with $\alpha = \nicefrac{1}{4}$, we get,
\[
\vol(B_{t+1} + z_i \B) \le (1-6^{-d})\vol(B_{t} + z_i \B).
\]
Thus, if the price is accepted, the contribution of the above term to the decrease of the potential gives,
\[
P(t+1)-P(t) \leq 64\cdot 6^d \cdot z_i\log(1-6^{-d}) \leq -2w,
\]
where we used $\forall x\in (0,1]: \log(1-x)\leq -x$, and $8 \cdot z_i \geq w$.
\end{proof}

\section{The Contextual Case: Profit Maximization}
\label{sec:profit}

In this section, we present our algorithms for profit maximization. We present our near-optimal results for
two-bit feedback in \Cref{subsec:two-bit-profit}, and then adapt it to
one-bit feedback in \Cref{subsec:one-to-two-bit-profit}. Finally, we
give the one-bit algorithm with per-round budget balance in
\Cref{subsec:one-bit-profit}.

    \subsection{Profit Maximization with Two-Bit Feedback}
    \label{subsec:two-bit-profit}

    We now build on the recursive zooming idea from
    \Cref{sec:cont-free-profit-maximization}, combined with the efficiency-maximization
    from \Cref{subsec:two-bit-gft}, to get our profit maximization
    algorithm for two-bit feedback. We report here a high-level description, while the pseudocode is reported in \Cref{app:pseudocodes-profit} as \Cref{alg:2bit-profit}.

    The algorithm maintains the confidence regions
    $S_t$ and $B_t$ and a family of Steiner Potentials for scales
    $z_i = \tfrac{2^{-3 \cdot 2^i}}{16d}$. Every time a new context
    $x_t$ arrives, it proceeds by cases, according to the widths of
    $S_t$ and $B_t$. (i) If both widths are smaller than
    $\nicefrac 1T$, then the error in estimating the actual valuations
    is sufficiently small, so the algorithm posts safe prices to both
    agents. (ii) If $\width(S_t,x_t)\ge \width(B_t,x_t)$, then the
    algorithm posts a safe price for the buyer, and $m^S_t + z_{i_t}$
    to the seller, where $i_t$ is the largest index such that the
    width of $S_t$ is upper bounded by $2^{-2^{i_t}}$ and $m^S_t$
    solves
        \begin{equation}
            \label{eq:unbalanced_seller}
            \vol(\{v \in S_t + z_{i_t}\B \mid \ip{v,x_t} \ge m^S_t\}) = 2^{-2^{{i_t}-1}}\vol(S_t + z_{i_t}\B).
        \end{equation}
        We are left with the last case (iii), where the width of the buyer dominates that of the seller, and it is larger than $\nicefrac 1T$. There, the algorithm posts the safe price to the seller, and $m^B_t-z_{i_t}$ to the buyer, where $i_t$ is the largest index such that the width of $B_t$ is upper bounded by $2^{-2^i}$, and $m^B_t$ solves the following equality:
        \begin{equation}
            \label{eq:unbalanced_buyer}
            \vol(\{v \in B_t + z_{i_t}\B \mid \ip{v,x_t} \le m^B_t\}) = 2^{-2^{{i_t}-1}}\vol(B_t + z_{i_t}\B).
        \end{equation}
        We call the prices induced by solving \Cref{eq:unbalanced_seller,eq:unbalanced_buyer} as \emph{unbalanced} prices. 
        In all three cases, the price for the seller may be larger than that for the buyer. If this happens, the algorithm maintains per-round budget balance by actually posting to both agents the unbalanced price, or any price in case (i) when no unbalanced price is used. 

        Unbalanced prices share the desirable property of drastically reducing the corresponding potential when rejected, while they reduce it by a smaller factor when they are accepted. We formalize this fact in the following two lemmata, adapted from \citet{LiuLS21}, and whose proofs are deferred to \Cref{app:profit}.

            \begin{restatable}{lemma}{lemrefuse}{\textnormal{(Refusing Unbalanced Prices)}}\label{lem:profit-2bit-refuse}
           If $\width(x_t, S_t) \ge \nicefrac 1T$ and the seller refuses $m^S_t + z_{i_t}$, then 
           \[
            \vol(S_{t+1} + z_{i_t}\B) \le 2^{-2^{i_t-1}} \vol(S_t+z_{i_t}\B).
           \]
           Similarly, if $\width(x_t, B_t) \ge \nicefrac 1T$ and the buyer refuses price $m^B_t - z_{i_t}$, then 
           \[
            \vol(B_{t+1} + z_{i_t}\B) \le 2^{-2^{i_t-1}} \vol(B_t+z_{i_t}\B).
           \]
        \end{restatable}

        \begin{restatable}{lemma}{lemaccept}{\textnormal{(Accepting Unbalanced Prices)}}\label{lem:profit-2bit-accept}
            If $\width(x_t, S_t) \ge \nicefrac{1}{T}$, and the seller accepts price $m^S_t+z_{i_t}$, then 
            \[
                \vol(S_{t+1} + z_{i_t}\B) \le \left(1-\frac{1}{10\cdot 2^{2^{(i_t-1)}}}\right)\vol(S_{t} + z_{i_t}\B).
            \]
            Similarly, if $\width(x_t, B_t) \ge \nicefrac{1}{T}$, and the buyer accepts price $m^B_t-z_{i_t}$, then
            \[
                \vol(B_{t+1} + z_{i_t}\B) \le \left(1-\frac{1}{10\cdot 2^{2^{(i_t-1)}}}\right)\vol(B_{t} + z_{i_t}\B).
            \]
        \end{restatable}

        We have all the ingredients to prove the regret guarantees of the two-bit algorithm.

        \begin{theorem}
        \label{thm:2-bit-profit}
            Consider the contextual bilateral trade problem in the two-bit feedback model. There exists a profit-maximizing algorithm that achieves regret $O(d \log \log T + d \log d)$ and enforces per-round budget balance. 
        \end{theorem}
        \begin{proof}
            \Cref{alg:2bit-profit} respects per-round budget balance since it always posts a lower price to the seller than it does to the buyer. Moreover, we can ignore from the analysis all time steps in which $s_t > b_t$, as they do not influence the profit of the algorithm and the benchmark. Likewise, we can disregard all time steps such that $\max\{\width(S_t,x_t),\width(B_t,x_t)\} \le \nicefrac 1T$ since they contribute $O(1)$ to the overall regret.

            We now analyze the regret incurred by the algorithm in the seller/buyer dominating cases. Consider the seller dominating case, as the buyer dominating case proceeds identically, but with signs flipped. We have two subcases: If $m^S_t + z_{i_t} < \ip{s,x_t}$, so that the seller refuses this price, the loss in profit is at most $1$, and, by \Cref{lem:profit-2bit-refuse}, 
            \[
                \vol(S_{t+1} + z_{i_t}\B) \le 2^{-2^{i_t-1}} \vol(S_t+z_{i_t}\B).
            \]
            Then, identically to the proof of \Cref{thm:2-bit-gft}, we know that $z_{i}^d \vol(\B) \le \vol(S_t + z_{i}\B) \le 3^d \vol(\B)$ for all $i$, then the seller cannot refuse a price corresponding to index $i$ more than
            \[
                n_{i}^{S, \textsf{ref}} \leq \frac{d\log \nicefrac{3}{z_{i}}}{2^{i-1}} = 6d + \frac{d\log 48d}{2^{i-1}}.
            \]
            Conversely, if $m^S_t + z_{i_t} \ge \ip{s,x_t}$, so that the seller accepts this price, the loss in profit is much smaller than in the case of refusal, namely at most $2 \cdot 2^{-2^{i_t}}$. Also, by \Cref{lem:profit-2bit-accept}, 
            \[
                \vol(S_{t+1} + z_{i_t}\B) \le \left(1-\frac{1}{10\cdot 2^{2^{{(i_t-1)}}}}\right)\vol(S_{t} + z_{i_t}\B).
            \]
            Then, as above, the seller cannot accept a price corresponding to index $i$ more than
            \begin{align*}
                 n_{i}^{S, \textsf{acc}} &\le \frac{d\log \nicefrac{z_{i}}{3}}{\log \left(1-\frac{1}{10\cdot 2^{2^{(i-1)}}}\right)} \\
                 &\le 10d2^{2^{(i-1)}}\log (48d\cdot 2^{3 \cdot 2^{i}}) \tag{as $\log(1+r) \le r$ for $r > -1$}\\
                 &\le 10 d2^{2^{(i-1)}}(\log (48d)+ 3\cdot 2^{i}).
            \end{align*}
            Finally, note that $i_t \leq 2\log\log T$, and hence, the accrued regret is at most
            \begin{align}
                \sum_{i=0}^{2\log\log T} \left(n_i^{S, \textsf{ref}} + n_i^{S, \textsf{acc}} \cdot 2^{-2^i}\right) &\le \sum_{i=0}^{2\log\log T} \left[6d + \frac{d\log 48d}{2^{i-1}} + 10 d2^{2^{{-(i-1)}}}(\log (48d)+ 3\cdot 2^{i})\right] \nonumber\\
                &\in O(d\log\log T + d\log d).
            \label{eq:final_profit}
            \end{align}
            We can perform a similar analysis for the time steps where the buyer's width dominates the seller's, thus concluding the proof.
        \end{proof}

        The positive result for profit maximization with two-bit feedback is essentially tight, given the lower bound for one-sided pricing from \citet{KleinbergL03}, which we can ``amplify'' using the $d$ directions (Proof deferred to \Cref{app:profit}).

        \begin{restatable}{proposition}{profitlower}
            \label{thm:2-bit-profit-lower}
            Consider the contextual bilateral trade problem in the two-bit feedback model. Any profit-maximizing algorithm suffers regret $\Omega(d \log \log \nicefrac Td)$.
        \end{restatable}

    \subsection{Profit Maximization with One-Bit Feedback}
    \label{subsec:one-to-two-bit-profit}

        We present how to modify the two-bit algorithm to work with only one bit of feedback. As in the efficiency maximization case, we only update the confidence regions when there is no ambiguity in the feedback, i.e., when one of the two agents is proposed a safe price. The tweak is fairly simple. When the two intervals $\its$ and $\itb$ do not intersect, or have both widths smaller than $\nicefrac 1T$, the one-bit algorithm behaves exactly as the two-bit one: At least one of the two agents receives a safe price in such a case, and therefore there is no ambiguity (or $\ls > \ub$, in which case the trade cannot happen and no confidence region is updated). If $\its$ and $\itb$ intersect, and the width is large enough, then the two-bit algorithm posts an unbalanced price to the agent with the largest width, \emph{and a safe price to the other}, to avoid ambiguity. Note that the two-bit feedback algorithm does something similar but caps the other agent's price to prevent violating the per-round budget balance. As we formally argue in the following theorem, we can remove this cap without incurring more than $O(d \log d)$ negative profit.

        \begin{theorem}
        \label{thm:profit-safe}
            Consider the contextual bilateral trade problem in the one-bit feedback model. There exists a profit-maximizing algorithm that achieves regret $O(d \log \log T + d \log d)$ and whose negative profit is at most $O(d \log d)$. 
        \end{theorem}
        \begin{proof}
            The analysis of the two-bit feedback carries exactly as it is, with the sole difference that we need to account for the negative profit incurred when posting non-budget-balanced prices. This happens only when the intersection of $\its$ and $\itb$ is non-empty, for a negative profit of order $O(2^{2^{-i}})$, where $i$ is the index corresponding to the largest width. Therefore, we can upper bound the negative profit as in \Cref{eq:final_profit}, with the only difference that when we upper bound the negative profit with $2^{2^{-i}}$, for an overall bound of $O(d \log d)$.
        \end{proof}
\subsection{Profit Maximization with One-Bit Feedback and Per-Round Budget Balance}\label{subsec:one-bit-profit}
We finally consider contextual profit-maximization under one-bit feedback and per-round budget balance and study Algorithm \ref{alg:1bit-profit-bb} whose pseudocode is given in Appendix~\ref{app:pseudocodes-profit}. Like the one-bit efficiency-maximizing algorithm analyzed in \Cref{subsec:one-bit-gft}, \Cref{alg:1bit-profit-bb} only updates the confidence regions $S_t$ and $B_t$ when there is no ambiguity in the (one-bit) feedback.  At each time step $t$, the algorithm computes the intervals $\its$ and $\itb$ and considers the following three cases:
        \begin{itemize}
        \item \emph{Small widths case:} if the two intervals have width less than $1/T$, then it posts a ``safe'' price for the seller, i.e., $p_t\gets \us$,  and a ``nearly safe'' price for the buyer, i.e., $q_t\gets \max\{\us,\lb\}$ (\Cref{fig:1bit-profit-a});
        \item \emph{Weak overlap case:} let $p_t$ be the middle-point of the larger interval, if it is accepted for sure by the other agent, then the algorithm posts it, \emph{middle price} (\Cref{fig:1bit-profit-b});
        \item \emph{Strong overlap case:} when neither of the two previous cases apply, the algorithm posts a price uniformly at random in $[-1,1]$, \emph{uniform price} (\Cref{fig:1bit-profit-c}).
        \end{itemize}
 Informally, the \textit{small widths case} replaces the \textit{well-separated} case of the efficiency-maximizing algorithm. The idea is that, while separation of $\proj(B_t,x_t)$ and $\proj(S_t,x_t)$ suffices to get full efficiency at round $t$, it does not suffice to extract a good profit at round $t$. 
 Instead, in the small widths case, we can choose $p_t$ and $q_t$ so that the regret in the profit does not exceed $2/T$.

The main result of this section is that \Cref{alg:1bit-profit-bb} achieves a regret of at most $O(d6^d \log T)$. In the proof of this result, we will use several lemmas and notations introduced in \Cref{subsec:one-bit-gft}, which we refer the reader to. 

\def\sellerbody{plot[smooth cycle] coordinates {(-1, 0.4) (0, 1.2) (1, 0.8) (0.8, -0.2) (-0.5, -0.3)}}
\def\buyerbody{plot[smooth cycle] coordinates {(-0.6, 0.3) (0, 0.8) (0.6, 0.5) (0.4, -0.1) (-0.3, -0.15)}}

\begin{figure}[t!]
    \centering
    \begin{subfigure}[b]{0.32\textwidth}
    \centering
        \begin{tikzpicture}[scale=0.8, every node/.style={scale=0.8}]
            \draw[->, gray] (0,-2.5) -- (0,1.2) node[left] {$x_t$};
            \draw[-, blue] (0,-1.39) -- (3.6,-1.39);
            \draw[-, green] (0,-0.22) -- (3.6,-0.22);
            
            \begin{scope}[shift={(1.5, -1.9)}, rotate=20, scale=0.4] 
                \draw[draw=yellow!60!black, fill=yellow!20, thick] \sellerbody;
                \node (s) at (0.1, 0.4) {$S_t$}; 
            \end{scope}
    
            \begin{scope}[shift={(1.5, -0.1)}, rotate=-10, scale=0.5] 
                \draw[draw=red!80!black, fill=red!20, thick] \buyerbody;
                \node (b) at (0, 0.2) {$B_t$}; 
            \end{scope}
    
            \fill (0,-1.39) circle (2pt) node[above right] {$p_t$};
            \fill (0,-0.22) circle (2pt) node[above right] {$q_t$};
        \end{tikzpicture}
        \caption{Small widths case}
        \label{fig:1bit-profit-a}
    \end{subfigure}
    \hfill
    \begin{subfigure}[b]{0.32\textwidth}
    \centering
        \begin{tikzpicture}[scale=0.8, every node/.style={scale=0.8}]
            \draw[->, gray] (0,-2.5) -- (0,1.2) node[left] {$x_t$};
            \draw[-, blue] (0,-1.1) -- (3.6,-1.1);
            
            \begin{scope}[shift={(1.2, -1.6)}, rotate=20] 
                \draw[draw=yellow!60!black, fill=yellow!20, thick] \sellerbody;
                \node (s_label_in_body) at (0.1, 0.4) {$S_t$}; 
            \end{scope}
    
            \begin{scope}[shift={(2.8, -0.75)}, rotate=-10] 
                \draw[draw=red!80!black, fill=red!20, thick] \buyerbody;
                \node (b_label_in_body) at (0, 0.2) {$B_t$}; 
            \end{scope}

    
            \fill (0,-1.1) circle (2pt) node[above right] {$p_t$};
        \end{tikzpicture}
        \caption{Weak overlap (middle price)}
        \label{fig:1bit-profit-b}
    \end{subfigure}
    \hfill
    \begin{subfigure}[b]{0.32\textwidth}
    \centering
        \begin{tikzpicture}[scale=0.8, every node/.style={scale=0.8}]
            \draw[->, gray] (0,-2.5) -- (0,1.2) node[left] {$x_t$};
            \draw[-, blue] (0,-0.65) -- (3.6,-0.65);
            
            \begin{scope}[shift={(1.2, -1.6)}, rotate=20] 
                \draw[draw=yellow!60!black, fill=yellow!20, thick] \sellerbody;
                \node (s_label_in_body) at (0.1, 0.4) {$S_t$};
            \end{scope}
    
            \begin{scope}[shift={(2.8, -1.2)}, rotate=-10] 
                \draw[draw=red!80!black, fill=red!20, thick] \buyerbody;
                \node (b_label_in_body) at (0, 0.2) {$B_t$}; 
            \end{scope}

    
            \fill (0,-0.65) circle (2pt) node[above right] {$p_t$};
        \end{tikzpicture}
        \caption{Strong overlap (random price)}
        \label{fig:1bit-profit-c}
    \end{subfigure}

    \caption{Visualization of the profit-maximizing algorithm with one-bit feedback and per-round budget balance.}
    \label{fig:1bit-profit}
\end{figure}
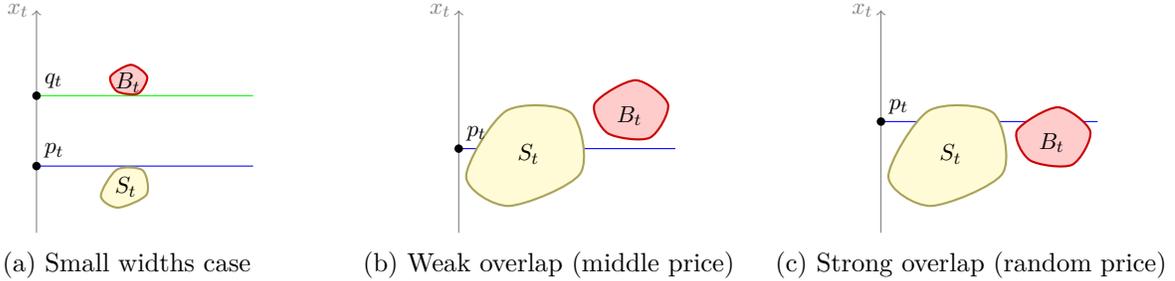

\begin{theorem}
\label{thm:one-bit-profit}
    Consider the contextual bilateral trade problem in the one-bit feedback model. There exists a profit-maximizing algorithm that achieves regret $O(d 6^d \log T)$ and enforces per-round budget balance. 
\end{theorem}

\begin{proof}
Consider the following potential, in which $z = \nicefrac{1}{16T}$,
\begin{equation}\label{eq:pot-profit}
P(t) = 2\cdot 6^d\left[\log \left(\frac{\vol(B_t + z\B)}{z^d\vol(\B)}\right)+\log \left(\frac{\vol(S_t + z\B)}{z^d\vol(\B)}\right)\right] + 2\frac{T-t}{T}.
\end{equation}
The potential satisfies that $\forall t: P(t)\geq 0$ and as well as,
\[
P(0) \le 6^d\cdot 4\cdot d\log(2/z)+ 2 \in O(d6^d\log T),
\]
where we used $\vol(B_t + z\B)\leq 2^d\vol(\B)$. As in the proof of \Cref{thm:one-bit-gft}, we aim to show that the following equation is always satisfied, where $R(t+1)$ denotes the regret of the algorithm after the $t$-th bid, and 
where $\mathbb{E}_t[\cdot]$ is the expectation conditioned on values taken by $p_1,\dots,p_{t-1}$,
\begin{equation}\label{eq:potential-eq2}
    \mathbb{E}_t[R(t+1)+P(t+1)]\leq R(t)+P(t).
\end{equation}
Equation \eqref{eq:potential-eq2} directly implies by induction $\forall t: R(t)\leq P(0)$ and thus the desired result. We refer to the proof of \Cref{thm:one-bit-gft} for more details on the formalism behind Equation \eqref{eq:potential-eq2}, and we turn to showing that Equation \eqref{eq:potential-eq2} holds with the potential defined in \eqref{eq:pot-profit}. We do the following case-disjunction.  

\paragraph{Small widths case.} If the width of both intervals $\proj(S_t,x_t)$ and $\proj(B_t,x_t)$ is less than $\nicefrac{1}{T}$, the algorithm posts a safe price for the seller $p_t = \us$, and posts price $q_t = \max\{p_t,\lb\}$ for the buyer. We argue that in this case, the regret is at most $2/T$. 

First, if $\us\leq \lb$, then $q_t = \lb$ and the trade is accepted with a regret of $(b_t-s_t)^+-(\us-\lb)^+  = (b_t-\lb)+(\us-s_t)\leq \nicefrac{2}{T}$. 
Second, if $\us\geq \lb$ and $b_t\geq \us$ the trade is accepted because $q_t = \us$, but it generates zero profit since $p_t = q_t$. The regret is $(b_t-s_t)^+ = b_t-s_t =  (b_t-\us)+(\us-s_t)\leq (b_t-\lb)+(\us-s_t)\leq \nicefrac{2}{T}$. 
Third, if $\us\geq \lb$ and $b_t\leq \us$, the trade is rejected and the regret is $(b_t-s_t)^+ \leq (\us-s_t)^+\leq 1/T$.

Thus, the increase in the regret satisfies  $R(t+1)-R(t)\leq \nicefrac{2}{T}$. The last term of the potential definition in \eqref{eq:pot-profit} implies that $P(t+1)-P(t)\leq -\nicefrac{2}{T}$, which completes the proof of Equation \eqref{eq:potential-eq2}.

\paragraph{Weak overlap case.} In the weak overlap case, assume, without loss of generality, that the width of the seller is larger, so that $\frac{\us+\ls}{2}\leq \lb$. We naively bound the increase in regret by $R(t+1)-R(t)\leq 2$. Since we have that $z\leq \width(S_t)/2$, we apply Lemma \ref{lem:unbalanced} to obtain that $\vol(S_{t+1} + z \B) \le (1-4^{-d})\vol(S_{t} + z \B).$
Therefore, taking the logarithm, we get
\[
    P(t+1)-P(t)\leq 2\cdot 6^d\log(1-4^{-d})\leq -2,
\]
where we used $\log(1-x)\leq -x$, which implies \eqref{eq:potential-eq2}.

\paragraph{Strong overlap case.} In the case of a strong overlap, we bound the increase of the regret by $R(t+1)-R(t)\leq (b^Tx_t-s^Tx_t)^+$. We then note that the probability that the price selected uniformly at random in $[-1,1]$ is accepted satisfies 
\[
    \mathbb{P}_t[p_{t}~\text{accepted}]\geq (b^Tx_t-s^Tx_t)^+/2.
\]
Thus, to prove Equation \eqref{eq:potential-eq2}, it suffices to show the following bound on the conditional expectation
\[
\mathbb{E}_t[P(t+1)-P(t) |p_{t}~\text{accepted}]\leq -2.
\]
We will show the stronger claim, that for any event in which the price is accepted, the bound $P(t+1)-P(t)\leq -2$ applies. For this, we consider two sub-cases.

\paragraph{Sub-case 1: $\width(B_t,x_t)$ and $\width(S_t,x_t)$ are within a factor $4$.} In this case, since at least one of the two quantities is greater than $\nicefrac{1}{T}$, we have that $\nicefrac{1}{16T} =z \leq \nicefrac{1}{4}\min\{\width(B_t,x_t),\width(S_t,x_t)\}$. By \Cref{lem:balanced-1bit}, when the price is accepted, at least one of the following two inequalities is verified: 
        \[
        \begin{cases}
            \vol(S_{t+1} + z \B) \le (1-6^{-d})\vol(S_{t} + z \B), \quad &\forall z \text{ such that } z \le \nicefrac{1}{4}\width(S_t,x_t)    \\
            \vol(B_{t+1} + z \B) \le (1-6^{-d})\vol(B_{t} + z \B), \quad &\forall z \text{ such that } z \le \nicefrac{1}{4}\width(B_t,x_t)
        \end{cases}
        \]
Therefore, the potential, which uses $z=\nicefrac{1}{16T}$, satisfies
\[
P(t+1)-P(t)\leq 2\cdot 6^d\cdot \log(1-6^{-d}) \leq -2,
\]
where we used $\forall x\in (0,1]: \log(1-x)\leq -x$, completing the analysis of this sub-case. 

\paragraph{Sub-case 2: $\width(B_t,x_t)\geq 4\width(S_t,x_t)$ or $\width(S_t,x_t)\geq 4\width(B_t,x_t)$.} We assume without loss of generality that $\width(B_t,x_t)\geq 4\width(S_t,x_t)$ and we treat the other case analogously. Notice that, in this case, an accepted price must belong to the last three-quarters of $B_t$, and that $z\leq \nicefrac{1}{4}\width(B_t,x_t)$ so that, applying Lemma \ref{lem:partition} with $\alpha = \nicefrac{1}{4}$, we get,
\[
\vol(B_{t+1} + z \B) \le (1-6^{-d})\vol(B_{t} + z \B).
\]
Thus, the decrease in the potential satisfies
\[
P(t+1)-P(t) \leq 2\cdot 6^d \log(1-6^{-d}) \leq -2,
\]
where we used $\forall x\in (0,1]: \log(1-x)\leq -x$.
\end{proof}

{\small
\bibliographystyle{plainnat}
\bibliography{references}
}

\newpage
\noindent {\bf \LARGE Appendix}

\appendix
\section{Pseudocodes for efficiency maximization}\label{app:pseudocodes-efficiency}


\begin{algorithm}
        \caption{Efficiency Maximization with Two-Bit Feedback}
        \begin{algorithmic}[1]
        \State $S_1, B_1 \gets \B$, $z_i = \nicefrac{2^{-i}}{8d}$ for all $i$ \Comment{Initial confidence sets and scale parameters}
        \For{$t=1, \dots, T$}
            \State The adversary reveals the context $x_t$
            \State let $i_t$ be the largest index $i$ such that $\max\{\width(S_{t},x_t),\width(B_{t},x_t)\}\le 2^{-i}$
            \If{$\its \cap \itb= \emptyset$}\Comment{Well-separated}
                \State $p_t \gets \us$ 
            \ElsIf{$\width(S_{t},x_t) \ge  \width(B_{t},x_t)$} \Comment{Seller Dominating}
                \State $p_t \gets$ the solution of $ \vol(\{v \in S_t + z_{i_t}\B \mid \ip{v,x_t} \le p_t\}) = \tfrac{1}{2} \vol(S_t + z_{i_t} \B)$
            \Else \Comment{Buyer Dominating}
                \State $p_t\gets$ the solution of $ \vol(\{v \in B_t + z_{i_t}\B \mid \ip{v,x_t} \ge p_t\}) = \tfrac{1}{2} \vol(B_t + z_{i_t} \B)$
            \EndIf
            \State Post price $p_t$, and update $S_t$ and $B_t$ according to $\ind{\ip{s,x_t} \le p_t}$ and $\ind{\ip{b,x_t} \ge p_t}$
        \EndFor
        \end{algorithmic}
        \label{alg:2bit-gft}
        \end{algorithm}
        
\begin{algorithm}
        \caption{Efficiency Maximization with One-Bit Feedback and Per-Round Budget Balance}
        \begin{algorithmic}[1]
        \State $S_1, B_1 \gets \B$ \Comment{Initial confidence sets}
        \For{$t=1, \dots, T$}
            \State The adversary reveals the context $x_t$
            \If{$\us\leq \lb$}\Comment{Well-separated}
            \State $p_t \gets \us$             \ElsIf{$\width(S_{t},x_t) \ge  \width(B_{t},x_t)$ and $\nicefrac{(\ls + \us)}{2} \le \lb$}\Comment{Weak overlap}
            \State $p_t\gets \nicefrac{(\ls + \us)}{2}$
            \ElsIf{$\width(B_{t},x_t) \ge  \width(S_{t},x_t)$ and $\nicefrac{(\lb + \ub)}{2} \geq \us$}\Comment{Weak overlap}
            \State $p_t\gets \nicefrac{(\lb + \ub)}{2}$
            \Else \Comment{Strong overlap}
            \State $p_t \gets$  a price drawn u.a.r. in $\its \cup \itb$ 
            \EndIf
        \State Post price $p_t$, and update $S_t$ and $B_t$ if there is no ambiguity in the feedback
        \EndFor
        \end{algorithmic}
        \label{alg:1bit-gft}
        \end{algorithm}

\newpage
\section{Pseudocodes for profit maximization}\label{app:pseudocodes-profit}
\begin{algorithm}[!ht]
        \caption{Profit Maximization with Two-Bit Feedback}
        \begin{algorithmic}[1]
        \State $S_1, B_1 \gets \B$, $z_i = \nicefrac{2^{-3 \cdot 2^i}}{16d}$ for all $i$ \Comment{Initial confidence sets and scale parameters}
        \For{$t=1, \dots, T$}
            \State The adversary reveals the context $x_t$
            \State let $i_t$ be the largest index such that $\max\{\width(S_{t},x_t),\width(B_{t},x_t) \}\le 2^{-2^{i_t}}$
            \If{$\max\{\width(S_t,x_t),\width(B_t,x_t)\} \le \nicefrac 1T$}
                \State $p_t \gets \us$, $q_t \gets \max\{p_t,\lb\}$
            \ElsIf{$\width(S_{t},x_t) \ge  \width(B_{t},x_t)$} \Comment{Seller Dominating case}
                \State $m^S_t$ solves $ \vol(\{v \in S_t + z_{i_t}\B \mid \ip{v,x_t} \ge m^S_t\}) = 2^{-2^{{i_t}-1}}\vol(S_t + z_{i_t}\B)$ 
                \State $p_t \gets m^S_t+z_{i_t}$, $q_t \gets \max\{p_t,\lb\}$
            \Else \Comment{Buyer Dominating case}
                \State
                $m^B_t$ solves $ \vol(\{v \in B_t + z_{i_t}\B \mid \ip{v,x_t} \le m^B_t\}) = 2^{-2^{{i_t}-1}}\vol(B_t + z_{i_t}\B)$ 
                \State $q_t \gets m^B_t-z_{i_t}$, $p_t \gets \min\{q_t,\us\}$
            \EndIf
            \State Post price $p_t$ to the seller and $q_t$ to the buyer
            \State Update $S_t$ and $B_t$ according to $\ind{\ip{s,x_t} \le p_t}$ and $\ind{\ip{b,x_t} \ge q_t}$
        \EndFor
        \end{algorithmic}
        \label{alg:2bit-profit}
        \end{algorithm}


\begin{algorithm}
        \caption{Profit Maximization with One-Bit Feedback and Per-Round Budget Balance}
        \begin{algorithmic}[1]
        \State $S_1, B_1 \gets \B$\Comment{Initial confidence sets}
        \For{$t=1, \dots, T$}
            \State The adversary reveals the context $x_t$
            \If{$\max\{\width(S_t,x_t),\width(B_t,x_t)\} \le \nicefrac 1T$} \Comment{Small widths}
                \State $p_t \gets \us$, $q_t \gets \max\{p_t,\lb\}$
            \ElsIf{$\width(S_{t},x_t) \ge  \width(B_{t},x_t)$ and $\nicefrac{(\ls + \us)}{2} \le \lb$}\Comment{Weak overlap}
            \State $p_t \gets \nicefrac{(\ls + \us)}{2}$, $q_t \gets p_t$
            \ElsIf{$\width(B_{t},x_t) \ge  \width(S_{t},x_t)$ and $\nicefrac{(\lb + \ub)}{2} \geq \us$}\Comment{Weak overlap}
            \State $q_t \gets \nicefrac{(\lb + \ub)}{2}$, $p_t \gets p_t$
            \Else \Comment{Strong overlap}
            \State  Let $u_t$ be a price drawn u.a.r. in $[-1,1]$        
            \State $p_t \gets u_t$, $q_t \gets u_t$
            \EndIf
            \State Post price $p_t$ to the seller and $q_t$ to the buyer
            \State Update $S_t$ and $B_t$ if there is no ambiguity in the feedback
        \EndFor
        \end{algorithmic}
        \label{alg:1bit-profit-bb}
        \end{algorithm}

\newpage

\section[Omitted content from Section 4]{Omitted Content from \Cref{sec:efficiency}}\label{app:efficiency}

\lembalanced*

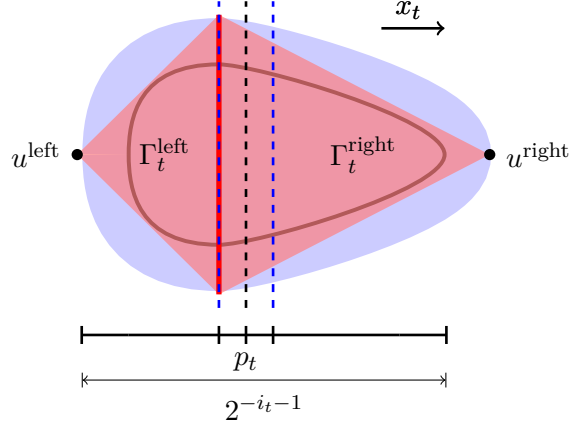
\begin{figure}[H]
\centering
\begin{tikzpicture}[scale=1.2]
  \newcommand{\pathA}{(0,1) .. controls (0,1.4) and (0.1, 2) .. (1,2)
              .. controls (1.4,2) and (3.5,1.4) .. (3.5,1)
              .. controls (3.5,0.6) and (1.4,0) .. (1,0)
              .. controls (0.1,0) and (0,.6) .. (0,1)}

 \begin{scope}
  \draw[line width=35pt, color=blue!20!white] \pathA;
  \draw[line width=1.5pt, fill=blue!20!white] \pathA;

  \fill[line width=0pt, fill=red!50!white, fill opacity=0.7] (-.57,1) -- (1,2.55) -- (1,-.55)  -- cycle;

  \fill[line width=0pt, fill=red!50!white, fill opacity=0.7] (4,1) -- (1,2.55) -- (1,-.55)  -- cycle;

  \begin{scope}[line width=1.0pt]
 \begin{scope}[>=latex]
 \draw[-] (0,-1)--(3,-1);
 \draw[-] (-.52,-1)--(0,-1);
 \draw[-] (3.52,-1)--(3,-1);
 \end{scope}
 \draw (-.52,-.9)--(-.52,-1.1);
 \draw (3.52,-.9)--(3.52,-1.1);
 \draw (1,-.9)--(1,-1.1);
\draw (1.6,-.9)--(1.6,-1.1);
\draw (1.3,-.9)--(1.3,-1.1);
\draw[->] (2.8,2.4)--(3.5,2.4);
 \end{scope}

 \node at (3.1, 2.6) {$x_t$};

 \draw[|<->|] (-.52,-1.5) -- node[below=2pt] {$2^{-i_t-1}$} (3.52,-1.5);
 
 \node at (1.3, -1.3) {$p_t$};
  
 \node at (0.4, 1) {$\Gamma_t^{\text{left}}$};
 \node at (2.6, 1) {$\Gamma_t^{\text{right}}$};

 \draw[line width=2pt, color=red] (1,2.55)--(1,-.55);

  \draw[line width=1pt, dashed, blue] (1,-.7)--(1,2.7);
  \draw[line width=1pt, dashed, blue] (1.6,-.7)--(1.6,2.7);
  \draw[line width=1pt, dashed] (1.3,-.7)--(1.3,2.7);

\node at (3.1, 2.6) {$x_t$};

  \node [shape=circle, fill=black,inner sep=1.5pt,label=left:$u^{\text{left}}$] (-.57,1) at (-.57,1) {};

 \node [shape=circle, fill=black,inner sep=1.5pt,label=right:$u^{\text{right}}$] (3.57,1) at (4,1) {};

\end{scope}
\end{tikzpicture}

\caption{Illustration for the proof of Lemma \ref{lem:balanced}: The blue opaque body is $S_t + z_{i_t}\B$ and the black solid contoured one is $S_t$. In faded red, the cones encapsulating the convex body. The solid red vertical line is the largest cross-section of $(d-1)$-dimensional volume $\nu$.}
\label{fig:proof}
\end{figure}

\begin{proof}
            We prove the statement concerning an accepted seller price, the other three cases are analogous (see \Cref{fig:proof} for an illustration). The seller's confidence region at time $t+1$ reads
            \(
                S_{t+1} = \{v \in S_t \mid \ip{ x_t,v} \leq p_t\}.
            \)
            
            Consider the subset of $S_{t+1} + z_{i_t}\B$ intersected with the half-space $\{ v :
            \ip{x_t,v} \leq p_t\}$. 
            By definition of $p_t$, and the fact that the seller accepts the price, we have that
            \begin{equation}
                \label{eq:S_t+1vsS_t}
                \vol(\{v \in S_{t+1} + z_{i_t}\B \mid \ip{v,x_t} \le p_t\}) \leq \vol(\{v \in S_t + z_{i_t}\B \mid \ip{v,x_t} \le p_t\}) = \tfrac{1}{2} \vol(S_t + z_{i_t} \B).
            \end{equation}
            Consider now the family of all cross-sections of $S_{t} + z_{i_t}\B$ orthogonal to direction $x_t$ that are at distance at most $z_i$ from the subspace $\{v: \ip{x_t,v}= p_t\}$, and denote with $\nu$ the $(d-1)$-dimensional volume of the largest such cross-section. The set $\{v \in S_{t+1} + z_{i_t}\B \mid \ip{v,x_t} \ge p_t\}$ is contained in a cylinder of width $z_i$ and of cross-section $\nu$, therefore:
            \begin{equation}\label{eq:S_t_part2}
               \vol(\{v \in S_{t+1} + z_{i_t}\B \mid \ip{v,x_t} \ge p_t\}) \leq \nu z_{i_t}. 
            \end{equation}
 Combining \Cref{eq:S_t+1vsS_t} with \Cref{eq:S_t_part2}, we thus have that
            \[
    \vol(S_{t+1}+z_{i_t}\B)\leq \tfrac{1}{2} \vol(S_t + z_{i_t} \B) + \nu z_{i_t}.
            \]
            Therefore, to conclude the proof, it suffices to show that $\nu z_{i_t}\leq \tfrac{1}{4} \vol(S_t + z_{i_t} \B)$. 
            
            Denote with $u^{\text{left}}_t$ and $u^{\text{right}}_t$ the extremes of $S_{t} + z_{i_t}\B$ along direction $x_t$; we introduce an orientation ``left'' and ``right'', in the sense that $u^{\text{left}}_t$ is the point corresponding to the smallest value of $\ip{v,s}$ for $v \in S_{t+1} + z_{i_t}\B$, while $u^{\text{right}}_t$ corresponds to the largest. 
            We then consider the cone $\Gamma^{\text{left}}_t$ (resp. $\Gamma^{\text{right}}_t$)  formed by taking the convex hull of the largest cross-section of $(d-1)$-dimension $\nu$ and $u^{\text{left}}_t$ (resp. $u^{\text{right}}_t$). Since the distance of $u^{\text{left}}_t$ and $u^{\text{right}}_t$ along direction $x_t$ is at least $2^{-i_t-1}$, and since $z_{i_t} = \nicefrac{2^{-i_t}}{8d}$, applying \Cref{fct:cone-volume}, we can conclude with
            \[
                \vol(S_t + z_{i_t} \B) \geq  \vol(\Gamma^{\text{left}}_t) + \vol(\Gamma^{\text{right}}_t) \geq \frac{2^{-i_t-1}\nu}{d} = 4\nu z_{i_t}. 
            \]
        \end{proof}

        \lowergft*
    \begin{proof}
    We prove this lower bound via Yao's Principle, by constructing a randomized instance such that any deterministic algorithm suffers, in expectation, the desired regret rate $\Omega(d)$.
    Consider an instance where the contexts are just the canonical basis vectors, arriving one after the other $e_1, e_2, \dots, e_d$. The instance can be arbitrarily long, but we only need these first $d$ contexts. We construct $s$ and $b$ independently per coordinate. Namely, for each coordinate $j \in \{1,\ldots, d\}$:
        \[
            (s_j, b_j) = \begin{cases}
                (0, \nicefrac{1}{3}), \text{ w.p. } \nicefrac{1}{2}\\
                (\nicefrac{2}{3}, 1), \text{ w.p. } \nicefrac{1}{2}
            \end{cases}.
        \]
        We note that, since every new context is not spanned by the previous ones and $(s_j,b_j)$ are determined independently per coordinate, the information provided by context $e_j$ is of no use to the algorithm for the successive contexts $e_{i}$ for $i>j$. 
        
        Consider now any learning algorithm, the problem faced at each new context is independent from the past observation, in particular for any price $p_t$ it posts, it either observes a trade with probability $\nicefrac 12$ (i.e., by posting prices in $[0,\nicefrac 13]\cup[\nicefrac 23,1]$) or it cannot observes a trade (for $p_t \in (\nicefrac 13,\nicefrac 23)$). Overall, its expected gain from trade is at most $\nicefrac d6$. Conversely, the benchmark attains a gain from trade of $\nicefrac{d}{3}$, as it always picks the right price between $\nicefrac 13$ and $\nicefrac 23$ for a per-round gain from trade of $\nicefrac 13$. This concludes the proof.
    \end{proof}

\lempartition*

We first establish the following (arguably folklore) facts:

\begin{fact}\label{fct:vol-scaling}
    Consider a body $K \in \R^d$ whose intrinsic $d$-dimensional volume is denoted by $\vol(K)$. Then, the body $K^\prime = \gamma K$ scaled by some positive scalar $\gamma > 0$ has volume $\gamma^d \cdot \vol(K)$.
\end{fact}
\begin{proof}
    We know that, for an injective, continuously differentiable map $f: U \to \R^d$, it holds that 
       \[
        \int_{f(U)} g(y) dy = \int_U g(f(x)) \cdot |\det(Df(x))| dx,
       \]  
    where $\det(Df)$ is the Jacobian determinant. We let $f(x)=\gamma x$ so that \(Df(x)=\gamma I_d\), and hence, the above becomes  
   \[
     \vol(\gamma K) = \int_K |\det(\gamma I_d)| dx = \gamma^d \cdot \int_K 1 dx = \gamma^d \cdot \vol(K),
   \]  
   as desired.
\end{proof}
\begin{fact}\label{fct:cone-volume}
    Consider a $d$-dimensional cone $\Gamma$, defined as the convex hull of a $d-1$-dimensional convex set $\Psi_{d-1}$ called the base, and a point $q$ called the apex, at distance $\ell$ from the subspace that supports $\Psi_{d-1}$.   
    Its volume satisfies
    \[
        \vol(\Gamma) = \frac{\ell}{d} \cdot \vol(\Psi_{d-1}),
    \]
    where 
    $\vol(\Psi_{d-1})$ is to be intended in dimension $(d-1)$.
\end{fact}
\begin{proof}
    Any cross‐section of the cone taken by a hyperplane parallel to the base at a distance \(x\) from the apex $q$ (with \(0 \le x \le \ell\)) is a scaled copy of \(\Psi_{d-1}\) with scaling factor $\nicefrac{x}{\ell}$.
    Thus, the \((d-1)\)-volume of this slice is
    \[
    \vol\left(\frac{x}{\ell} \cdot \Psi_{d-1}\right) = \left(\frac{x}{\ell}\right)^{d-1} \cdot \vol(\Psi_{d-1}),
    \]
    by \Cref{fct:vol-scaling}.
    We integrate over \(x\) from \(0\) to \(\ell\) to obtain a total volume of
    \[
    \vol(\Gamma) = \int_{0}^{\ell} \left(\frac{x}{\ell}\right)^{d-1}\vol(\Psi_{d-1})\,dx 
    = \vol(\Psi_{d-1}) \cdot \ell^{-d+1} \int_{0}^{\ell} x^{d-1} dx = \frac{\ell}{d} \cdot \vol(\Psi_{d-1}).
    \]
    This establishes the statement.
\end{proof}

\begin{figure}[h]
\centering
\begin{tikzpicture}[scale=1.2]
  \newcommand{\pathA}{(0,1) .. controls (0,1.4) and (.6, 2) .. (1,2)
              .. controls (1.4,2) and (3,1.4) .. (3,1)
              .. controls (3,0.6) and (1.4,0) .. (1,0)
              .. controls (.6,0) and (0,.6) .. (0,1)}
  \fill[blue!0!white] (-.1,-.4) rectangle (3.4,2.6);
  \draw[line width=40pt, color=blue!20!white] \pathA;
  \draw[line width=1.5pt, fill=blue!20!white] \pathA;

 \begin{scope}[line width=1.0pt]
 \begin{scope}[>=latex]
 \draw[-] (0,-1)--(3,-1);
 \draw[-] (-.52,-1)--(0,-1);
 \draw[-] (3.52,-1)--(3,-1);
 \end{scope}
 \draw (0,-.9)--(0,-1.1);
 \draw (-.52,-.9)--(-.52,-1.1);
 \draw (3.52,-.9)--(3.52,-1.1);
 \draw (3,-.9)--(3,-1.1);
 \draw (2,-.9)--(2,-1.1);
 \draw[->] (2.8,2.4)--(3.5,2.4);
 \end{scope}
\node at (3.1, 2.6) {$x$};

 \node at (-.23, -1.3) {$z$};
 \node at (1.7, 2.7) {$H$};

 \node at (3.23, -1.3) {$z$};
 \node at (1, -1.3) {$\alpha \ell$};
 \node at (2.5, -1.3) {$(1-\alpha) \ell$};

  \draw[line width=1pt, dashed] (2,-.7)--(2,2.7);

  \node [shape=circle, fill=black,inner sep=1.5pt,label=left:$u^{\text{left}}$] (-.57,1) at (-.57,1) {};
    \node [shape=circle, fill=black,inner sep=1.5pt,label=right:$u^\prime$] (1.43,1) at (1.43,1) {};

 \node [shape=circle, fill=black,inner sep=1.5pt,label=right:$u^{\text{right}}$] (3.57,1) at (3.57,1) {};


 \begin{scope}[xshift=6.2cm]
  \draw[line width=40pt, color=blue!20!white] \pathA;
  \draw[line width=1.5pt, fill=blue!20!white] \pathA;

   \draw[line width=1pt, dashed] (2,-.7)--(2,2.7);

  \fill[line width=0pt, fill=red!50!white, fill opacity=0.7] (-.57,1) -- (3.57,4) -- (3.57,-2)  -- cycle;
  

 \draw[line width=2pt, color=red] (1.43,2.5)--(1.43,-.5);

 \begin{scope}[line width=1.0pt]
 \begin{scope}[>=latex]
 \draw[-] (-.52,-1)--(2,-1);
 \end{scope}
 \draw (1.43,-.9)--(1.43,-1.1);
 \draw (-.52,-.9)--(-.52,-1.1);
 \draw (2,-.9)--(2,-1.1);
 \end{scope}

\node at (.5, -1.3) {$\alpha \ell$};
 \node at (1.64, -1.3) {$z$};

  \node [shape=circle, fill=black,inner sep=1.5pt,label=left:$u^{\text{left}}$] (-.57,1) at (-.57,1) {};
    \node [shape=circle, fill=black,inner sep=1.5pt,label=right:$u^\prime$] (1.43,1) at (1.43,1) {};

 \node [shape=circle, fill=black,inner sep=1.5pt,label=right:$u^{\text{right}}$] (3.57,1) at (3.57,1) {};

\end{scope}
\end{tikzpicture}

\caption{Illustration for the proof of Lemma \ref{lem:partition}.}
\label{fig:dlogd}
\end{figure}
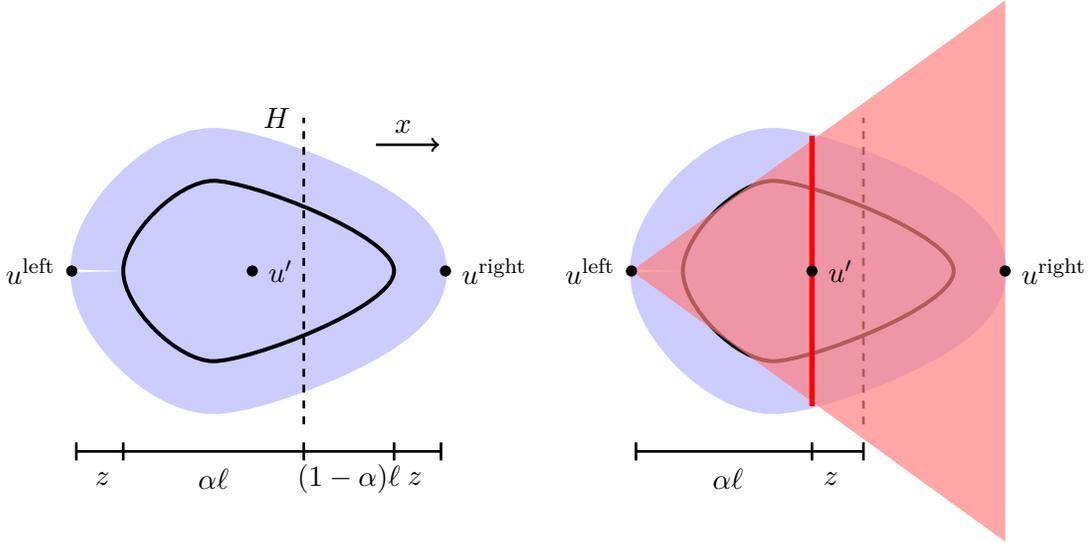
\FloatBarrier

\begin{proof}[Proof of \Cref{lem:partition}]
    We denote by $\ell = \width(K,x)$, the length of the segment $\proj(x, K)$. Without loss of generality, we assume that $H$ is of the form $H = \{y\in \mathbb{R}^d : \ip{x,y}\leq \lambda\}$, where $x$ is the vector normal to 
$h$, and that the convex set $K$ is closed and bounded. The notations introduced in the proof are illustrated in \Cref{fig:dlogd}.

We denote by $u^{\text{left}}$ (resp. $u^{\text{right}}$) a minimizer (resp. maximizer) of $:u\rightarrow\ip{u,x}$ in $K + z\B$. Note that we must have $\ip{u^{\text{right}}-u^{\text{left}},x} = \ell + 2z$. We also define the point $u'$ in the segment $[u^{\text{left}},u^{\text{right}}]$ satisfying $\ip{u'-u^{\text{left}},x} = \alpha \ell$. We denote by $h'$ the hyperplane orthogonal to $x$ passing through $u'$, and by $H'$ the corresponding half-space that contains $u^{\text{left}}$ (note that $H'\subset H$). We partition $K + z\B$ into two convex pieces according to $h'$, we denote by $L$ the part containing $u^{\text{left}}$ (i.e., $L = (K+z\B) \cap H'$), and by $R$ the part containing $u^{\text{right}}$ (i.e., $R = (K+z\B)\setminus H'$). 

Observing that $K\setminus H + z\B \subset R$ and that
    $\vol(K + z\B) = \vol(L) + \vol(R)$, we can write 
    \[
        \frac{\vol(K + z\B)}{\vol(K\setminus H + z\B)}\geq \frac{\vol(K + z\B)}{\vol(R)} = 1 +  \frac{\vol(L)}{\vol(R)}.
    \]
    Our goal is thus now to lower bound the numerator $\vol(L)$, and upper bound the denominator $\vol(R)$. To that end, consider the cone $\Gamma$ generated as the convex hull of the base $(K + z\B) \cap h'$ and the apex $u^{\text{left}}$. By convexity of $L$, we must have $\Gamma \subset L$, and thus $\vol(\Gamma)\leq \vol(L)$. Now further observe that, by convexity of $K+z\B$, $R$ must be contained in $\Gamma^+\setminus \Gamma$, where $\Gamma^+$ is a rescaled copy of $\Gamma$ with rescale factor $\nicefrac{\ell + 2z}{\alpha \ell}$. This implies that  $\vol(R)\leq  ((\nicefrac{\ell + 2z}{\alpha \ell})^d-1)\vol(\Gamma)$, which gives,
    \[
    \frac{\vol(K + z\B)}{\vol(K\setminus H + z\B)} \geq 1 + \frac{1}{(\nicefrac{\ell + 2z}{\alpha \ell})^d-1} = \frac{1}{1-\left(\frac{\alpha\ell}{\ell+2z}\right)^d} \geq \frac{1}{1-\left(\frac{\alpha}{1+2\alpha}\right)^d}
    \]
    since $z \leq \alpha\ell$. This concludes the proof.

\end{proof}

\section[Omitted content from Section 5]{Omitted Content from \Cref{sec:profit}}
\label{app:profit}

\lemrefuse*
        \begin{proof}
            Let us focus on the seller, as the proof for the buyer is identical, but with the signs flipped. Since the seller refuses the price, we know that
            \[
                \{v \in S_t + z_i\B \mid \ip{v,x_t} < m^S_t\} \cap (S_{t+1} + z_{i_t}\B) = \emptyset.
            \]
            Then, by definition of $m^S_t$, we have that
            \[
                \vol(S_{t+1} + z_{i_t}\B) \le 2^{-2^{i_t-1}} \vol(S_t+z_{i_t}\B),
           \]
            which concludes the proof.
        \end{proof}
        \lemaccept*
        \begin{proof}
            The proof proceeds similarly to that of \Cref{lem:balanced}. Again, we focus only on the seller's case, as the buyer's case holds identically but with signs flipped.

            First, observe that in case the seller accepts the price $m^S_t + z_{i_t}$, we have that 
            \[
                \{v \in S_t + z_{i_t}\B \mid \ip{v,x_t} \ge m^S_t + 2z_{i_t}\} \cap (S_{t+1} + z_{i_t}\B) = \emptyset,
            \]
            and hence,
            \[
                \vol(S_{t+1} + z_{i_t}\B) \le \vol(S_{t} + z_{i_t}\B) - \vol(\{v \in S_t + z_{i_t}\B \mid \ip{v,x_t} \ge m^S_t + 2z_{i_t}\}).
            \]
            Thus, our task reduces to bounding from below the volume of $\{v \in S_t + z_i\B \mid \ip{v,x_t} \ge m^S_t + 2z_{i_t}\}$. To this end, let $\nu$ be the volume of the largest cross-section inside $S_t + z_{i_t}\B$. We have that
            \[
                \vol(\{v \in S_t + z_{i_t}\B \mid m^S_t \le \ip{v,x_t} \le m^S_t + 2z_{i_t}\}) \leq 2\nu z_{i_t},
            \]
            since it has a cross-section of volume $\nu$ and a width of at most $2z_{i_t}$ in the $x_t$ direction. 
            By convexity of $S_t + z_{i_t}\B$ (and the same argument as in \Cref{lem:balanced}), we have that 
            \[
                \vol(S_t + z_{i_t} \B) \geq \frac{2^{-2^{i_t+1}}\nu}{d} \quad \Longrightarrow \quad \nu \le \frac{d}{2^{-2^{i_t+1}}} \vol(S_t + z_{i_t} \B). 
            \]
            The above, together with the definition of $m^S_t$, imply that
            \begin{align*}
                &\vol(\{v \in S_t + z_{i_t}\B \mid \ip{v,x_t}  \ge m^S_t + 2z_{i_t}\}) \\
                =~ &\vol(\{v \in S_t + z_{i_t}\B \mid \ip{v,x_t}  \ge m^S_t\}) - \vol(\{v \in S_t + z_{i_t}\B \mid m^S_t \le \ip{v,x_t} \le m^S_t + 2z_{i_t}\}) \\
                \ge~ &2^{-2^{{i_t}-1}}\vol(S_t + z_{i_t}\B) -  \frac{2dz_{i_t}}{2^{-2^{i_t+1}}} \vol(S_t + z_{i_t} \B) \\
                \ge~ &\frac{1}{10\cdot 2^{2^{{(i_t-1)}}}} \vol(S_t + z_{i_t} \B).
            \end{align*}
            This concludes the proof.
        \end{proof}

        \profitlower*
        \begin{proof}
            Consider the one-sided version of the problem, where the seller's valuation is always $0$. From \citet{KleinbergL03}, we know that the context-free version of such a problem admits a lower bound of $\Omega(\log \log T)$ on the regret for any profit-maximizing algorithm. We can amplify such a lower bound to $d$ dimensions by subdividing the time horizon into $d$ chunks of $\nicefrac Td$ time steps, and repeating the one-dimensional lower bound construction along the $d$ vectors of the canonical basis. In each chunk, any algorithm suffers $\Omega(\log \log \nicefrac Td)$ regret, for an overall regret of at least $\Omega(d\log \log \nicefrac Td)$.
        \end{proof}

\end{document}